%% file: flat.tex
\documentclass{lmcs}
\pdfoutput=1

\usepackage{lastpage}
\lmcsdoi{16}{4}{4}
\lmcsheading{}{\pageref{LastPage}}{}{}%
{Dec.~25,~2019}{Oct.~13,~2020}{}

\usepackage[utf8]{inputenc}

\usepackage{amsfonts,amstext,amsmath}
\usepackage[labelformat=simple]{subcaption}
\usepackage{mathrsfs}
\usepackage{extarrows,fontawesome}
\usepackage{url}
\usepackage{tikz}
\usepackage{booktabs}
\usepackage{multirow}
\usepackage[ruled]{algorithm2e}

\input{macros}
\input{tikzmacros}


\theoremstyle{remark}
\newtheorem*{problem}{Problem}

\begin{document}

\title{Verification of Flat FIFO Machines} 

\titlecomment{A preliminary version of this work appeared in the conference CONCUR 2019. The work reported was carried out in the framework of ReLaX, UMI2000 (ENS
Paris-Saclay, CNRS, Univ. Bordeaux, CMI, IMSc). This work was also supported
by the grant ANR--17--CE40--0028 of the French National Research Agency ANR
(project BRAVAS).}
\author{Alain Finkel\rsuper{{a,c}}}
\author{M.~Praveen\rsuper{{b,c}}}
\address{\lsuper{a}Universit\'{e} Paris-Saclay, ENS Paris-Saclay, CNRS, Laboratoire Sp\'{e}cification et V\'{e}rification, 91190, Gif-sur-Yvette, France and Institut Universitaire de France}
\address{\lsuper{b}Chennai Mathematical Institute, India}
\address{\lsuper{c}UMI ReLaX}

\thanks{Partially supported by a grant from the Infosys foundation.}

\maketitle

\begin{abstract}
The decidability and complexity of reachability problems and
model-checking for flat counter machines have been explored in detail.
However, only few results are known for flat (lossy) FIFO machines, only in
some particular cases (a single loop or a single bounded expression).
%
%
We prove, by establishing reductions between properties, and by
reducing SAT to a subset of these properties that many verification
problems like reachability, non-termination, unboundedness are \NP{}-complete for flat
FIFO machines, generalizing similar existing results for flat counter machines. We also show that reachability is \NP{}-complete for flat
lossy FIFO machines and for flat front-lossy FIFO machines.
%
%
%
We construct a trace-flattable system of many counter machines communicating via rendez-vous
that is bisimilar to a
given flat FIFO machine, which allows to model-check the original flat
FIFO machine. Our results lay the theoretical foundations and open the way to build a verification tool for (general) FIFO machines based on analysis of flat sub-machines.

\end{abstract}

\section{Introduction}

\input{introduction}

\section{Preliminaries}
\input{preliminaries}

\section{Complexity of Reachability Properties for Flat FIFO Machines}

\input{complexity}

\section{Complexity of Reachability for Flat Lossy FIFO Machines}
\input{front-lossy}

\section{Construction of an Equivalent Counter System}
\input{counter}

\section{Conclusion and Perspectives}
We answered the complexity of the main reachability problems for flat
(perfect, lossy and front-lossy) FIFO machines which are \NP{}-complete as for flat counter machines. We
also show how to translate a flat FIFO machine into a trace-flattable
counter system.
This opens the way to model-check a general FIFO machine
by \emph{enumerating its flat sub-machines}.

Let us recall the spirit of many tools for non-flat counter machines like FAST, FLATA, \ldots~\cite{FAST-cav03,BFLS05-atva,BFLP-sttt08,DBLP:conf/cade/BozgaIKV10,DFGD-jancl10} and for general well structured transition systems~\cite{FG-lmcs12}. The framework for underapproximating a non-flat machine $M$ proposes to enumerate a (potentially) infinite sequence of flat sub-machines $M_1, M_2,\ldots,M_n,\ldots$,
to compute the reachability set of each flat sub-machine $M_n$, and to iterate this process till the reachability set is computed. For this strategy, we use a fair enumeration of flat sub-machines, which means that every flat sub-machine will eventually appear in the enumeration.
%

Suppose $M_n$ is a flat FIFO sub-machine enumerated and we want to check if $Reach(M_n)$ is stable under $Post_M$. We don't want to compute directly $Reach(M_n)$ but we will compute $Reach(C_n)$ that is possible since $C_n$ is a flat counter machine. If there is a transition $t$ in the non-flat machine $M$ that does not have any copy in $M_n$ then $M_n$ is not stable and we continue. Otherwise, transition $t$ of $M$ has copies $t_1,\ldots,t_m$ in $M_n$. Check if $M_n$ is stable under each transition $t_1,\ldots,t_m$; this is done by testing whether, for every $i=1,\ldots,m$, $Post_{T_i}(Reach(C_n)) \subseteq Reach(C_n)$ where $T_i$ is the set of transitions, in $C_n$, associated (by bisimulation) with transition $t_i$ in $M_n$. If one of the $m$ tests fails, $M_n$ is not stable. Otherwise, $M_n$ is stable.

The following semi-algorithm gives an overview of a strategy to compute the reachability relation and then verify, for instance, whether a configuration is reachable from another one.
\begin{itemize}
\item  start fairly enumerating flat sub-machines $M_1, M_2,\ldots,M_n,\ldots$
\item   for every flat subsystem $M_n$
\begin{itemize}
\item   compute the synchronized counter system $C_n$ associated with $S_n$
\item compute the reachability set $Reach(C_n)$
\item test whether $Reach(M_n)$ is stable under $Post_M$
\item if $Reach(M_n)$ is stable under $Post_M$ we can terminate. Otherwise, we go to the next flat subsystem $M_{n+1}$ and repeat.
\end{itemize}
\end{itemize}
The above semi-algorithm terminates if there is a flat FIFO sub-machine having the same reachability set as the entire machine.

But real systems of FIFO systems are often not reduced to an \emph{unique} FIFO machine. Let us show how results on flat FIFO machines can be used to verify \emph{systems} of communicating FIFO machines.

%
Let us consider a peer-to-peer FIFO system $S=(M_1, M_2,\ldots,M_k)$ where machine $M_i$ communicates with machine $M_j$ through two one-directional FIFO channels: $M_i$ sends letters to $M_j$ through channel $c_{i,j}$ and  $M_i$ receives letters from $M_j$ through channel $c_{j,i}$ for every $i,j=1,\ldots,k$ ($i \neq j$). Remark that peer-to-peer flat FIFO systems don't produce (by product) a flat FIFO machine.
%
%
If we consider the product of the three flat  machines shown in Fig.~\ref{fig:motivation},
the resulting FIFO  machine is not flat. It does become flat if we
remove the self loop labeled $\mathtt{pq}?y$ in Process P. The resulting flat
sub-machine is unbounded, so it implies that the original system is also
unbounded. Hence, even if the given flat FIFO system don't produce (by product) a flat FIFO machine, some
questions can often be answered by analyzing sub-systems and flat sub-machines.
%
Fortunately, reachability in such peer-to-peer flat FIFO systems reduces to reachability in VASS~\cite{BDM-concur20}, hence the reachability problem is decidable (but with the non-elementary complexity of reachability in VASS).


When systems of FIFO machines $S=(M_1, M_2,\ldots,M_k)$ are not composed of flat FIFO machines, we may use different strategies: we may compute the product $M$ of machines $M_i$ and enumerate the flat sub-machines of $M$ or enumerate the flat sub-systems $S_n$ of $S$ and analyse them.

It remains to be seen if tools can be optimized to make verifying FIFO machines work in practice. This
strategy has worked well for counter machines and offers hope for FIFO
machines. We have to evaluate all these possible verification strategies on real case studies.

\paragraph{Acknowledgements.} We
would like to thank the anonymous referees of both the Conference CONCUR'2019 and the Journal LMCS
for their attentive reading and their constructive questions and suggestions that allowed us to improve the quality of our article.


\bibliographystyle{plain}
\bibliography{flat}
\vfill

%



\end{document}

%% file: macros.tex
\renewcommand{\vec}[1]{\mathbf{#1}}

\newcommand{\set}[1]{\{#1\}}

\newlength{\ml}

\newcommand{\alain}[1]{\todo[inline,color=red!20]{{\bf Al:} #1}}
\newcommand{\praveen}[1]{\todo[inline,color=blue!10]{{\bf P:} #1}}

\newcommand{\Nat}{\mathbb{N}}
\newcommand{\Int}{\mathbb{Z}}

\newcommand{\counters}{K}
\newcommand{\cstrans}[2]{\xlongrightarrow[#2]{#1}}
\newcommand{\cval}{\vec{\nu}}

\newcommand{\channels}{F}
\newcommand{\fstrans}[1]{\xlongrightarrow{#1}}

\newcommand{\fval}{\vec{w}}

\newcommand{\NP}{{\normalfont\textsc{Np}}}
\newcommand{\PSPACE}{{\normalfont\textsc{Pspace}}}
\newcommand{\EXPSPACE}{{\normalfont\textsc{Expspace}}}
\newcommand{\PTIME}{{\normalfont\textsc{Ptime}}}

\newcommand{\cnt}{\mathrm{count}}
\newcommand{\ord}{\mathrm{order}}
\newcommand{\sync}{\mathrm{sync}}
\newcommand{\simCont}{\mathrm{content}}
\newcommand{\flatcs}{\mathrm{flat}}

\newcommand{\ket}[1]{[#1\rangle} 

%% file: tikzmacros.tex
\usetikzlibrary{decorations.pathmorphing}
\usetikzlibrary{decorations.text}

\tikzstyle{state}=[draw=black, fill=black, circle, inner sep=0.01cm, minimum
height=0.1cm]
\tikzstyle{longpath}=[decorate,decoration={snake,amplitude=0.02cm,segment length=0.2cm}]
\tikzstyle{unreachable}=[draw=black!20,fill=black!20,text=black!20]

%% file: introduction.tex
\subsection*{FIFO machines}
%
%
Asynchronous distributed processes communicating through First In
First Out (FIFO) channels are used since the seventies as models for
protocols~\cite{DBLP:journals/sigops/Bochmann75}, distributed and
concurrent programming and more recently for web service choreography
interface~\cite{DBLP:conf/coordination/BusiGGLZ06}.
Since FIFO machines simulate counter machines, most reachability
properties are  \emph{undecidable} for FIFO machines: for example, the basic task of
checking if the number of messages buffered in a channel can grow
unboundedly is undecidable~\cite{DBLP:journals/jacm/BrandZ83}.

There aren't many interesting and useful FIFO
subclasses with a  \emph{decidable} reachability problem. Considering
FIFO machines with a unique FIFO channel is not a useful restriction
since they may simulate Turing machines~\cite{DBLP:journals/jacm/BrandZ83}. A few examples of decidable
subclasses are half-duplex systems~\cite{CF-icomp05} (but they are
restricted to two machines since the natural extension to three
machines leads to undecidability), existentially bounded deadlock free
FIFO machines~\cite{DBLP:journals/fuin/GenestKM07} (but it is
undecidable to check if a machine is existentially bounded, even for
deadlock free FIFO machines), synchronisable FIFO machines (the property
of synchronisability is undecidable~\cite{FL-icalp17} and moreover, it
is not clear which properties of synchronisable machines are
decidable), flat FIFO machines~\cite{DBLP:conf/sas/BoigelotGWW97,DBLP:journals/tcs/BouajjaniH99} and
lossy FIFO machines~\cite{DBLP:journals/fmsd/AbdullaCBJ04} (but one
loses the perfect FIFO mechanism).

\subsection*{Flat machine} A flat machine~\cite{BFLS05-atva,FG-lmcs12,DHLST-CONCUR17,DBLP:journals/sttt/Boigelot12}
is a machine with a single finite control structure such that every control-state belongs to at
most one loop. Equivalently, the language of the
control structure is included in a bounded language of the form
$w_1^*w_2^* \cdots w_k^*$ where every $w_i$ is a non empty word.
Analyzing flat machines essentially reduces to
accelerating loops (i.e., to compute finite representations of the
effect of iterating each loop arbitrarily many times) and to connect these finite
representations with one another.  Flat machines are particularly
interesting since one may under-approximate any machine by its flat
submachines.

For counter machines~\cite{DFGD-jancl10,DBLP:conf/atva/IosifS16}, this strategy lead to some tools
like FAST~\cite{FAST-cav03}, LASH, TREX~\cite{DBLP:conf/cav/2001},
FLATA~\cite{DBLP:conf/cade/BozgaIKV10} which enumerate all flat
submachines till the reachability set is reached. This strategy
is not an algorithm since it may never terminate on some inputs.
However in practice, it terminates in many cases; e.g., in~\cite{FAST-cav03}, 80\% of the examples (including Petri nets and
multi-threaded Java programs) could be effectively verified. The
complexity of flat counter machines is well-known: reachability is
\NP{}-complete for variations of flat counter machines~\cite{DBLP:phd/ethos/Haase12,DBLP:journals/corr/BozgaIK13,DBLP:journals/iandc/DemriDS15},
model-checking first-order formulae and linear $\mu$-calculus formulae is
\PSPACE{}-complete while model-checking B\"{u}chi automata is \NP{}-complete~\cite{DDS-icalp13}; equivalence between model-checking flat counter
machines and Presburger arithmetic is established in~\cite{DDS-tcs17}.

\subsection*{Flat FIFO machines} We know almost nothing about flat
FIFO machines, even the complexity of reachability is not known.
Boigelot et al.~\cite{DBLP:conf/sas/BoigelotGWW97} used recognizable languages (QDD) for accelerating loops in a subclass of flat FIFO machines, where there are restrictions on the number of channels that a loop can operate on.
Bouajjani and Habermehl~\cite{DBLP:journals/tcs/BouajjaniH99} proved that the acceleration of
\emph{any} loop can be finitely represented by combining a
deterministic flat finite automaton and a Presburger formula (CQDD)
that are both computable.  However, surprisingly, no upper
bound for the Boigelot et al.'s and for the Bouajjani et al.'s
loop-acceleration algorithms are known. Just the complexity of the
inclusion problem for QDD, CQDD and SLRE (SLRE are both QDD and CQDD)
are partially known (respectively \PSPACE-complete,
\textsc{N2Exptime}-hard,
\textsc{CoNp}-complete)~\cite{FPS-ICOMP}. But the  complexity of the
reachability problem for flat FIFO machines was not known. Only the
complexity of the control-state reachability problem was known to be
\NP-complete for single-path flat
FIFO machines~\cite{EGM2012}. Moreover, other properties and
model-checking have not been studied for flat FIFO machines. Similarly, Abdulla et al.'s studied the verification of lossy FIFO machines by accelerating loops and representing them by a class of regular expressions called Simple Regular Expressions (SRE)~\cite{DBLP:conf/cav/AbdullaBJ98,DBLP:journals/fmsd/AbdullaCBJ04} and gave a polynomial (quadratic) algorithm for computing  the reachability set $\sigma^*(L)$ of a loop labeled by $\sigma$ from a SRE language $L$. But the  complexity of the
reachability problem for flat lossy FIFO machines was not known.

\subsection*{Contributions}
We solve the open problem of the complexity of the reachability
problem for flat FIFO machines by showing that it is \NP{}-complete; we
extend this result to other usual verification properties and show
that they are also \NP{}-complete. We also show that the reachability
problem is \NP{}-complete for flat (front-)lossy FIFO machines. 
Then we show that a flat FIFO
machine can be simulated by a synchronized
product of counter machines.
This synchronized product is flattable and its reachability set is
semilinear.
%


%% file: preliminaries.tex
We write $\Int$ (resp.~$\Nat$) to denote the set of integers
(resp.~non-negative integers).
A finite alphabet is any finite set
$\Sigma$. Its elements are referred to as letters; $\Sigma^{*}$ is the
set of all finite sequences of letters, referred to as words. We
denote by $w_{1}w_{2}$ the word obtained by concatenating $w_{1}$ and
$w_2$; and
$\epsilon$ is the empty sequence, which is the unity for the
concatenation operation. We write $\Sigma^{+}$ for
$\Sigma^{*}\setminus \set{\epsilon}$. If $w_{1}$ is a prefix of
$w_{2}$, we denote by $w_{1}^{-1}w_{2}$ the word obtained from $w_{2}$
by dropping the prefix $w_{1}$. If $w_{1}$ is not a prefix of $w_{2}$,
then $w_{1}^{-1}w_{2}$ is undefined. A word $z \in \Sigma^{*}$ is
primitive if $z \notin w^{*}\setminus \set{w}$ for any $w \in
\Sigma^{*}$. We denote by $\mathit{Parikh}(w): \Sigma \to \Nat$ the
function that maps each letter $a \in \Sigma$ to the number of times
$a$ occurs in $w$. We denote by $w^{n}$ the concatenation of $n$ copies of
$w$. The infinite word $x^{\omega}$ is obtained by
concatenating $x$ infinitely many times.

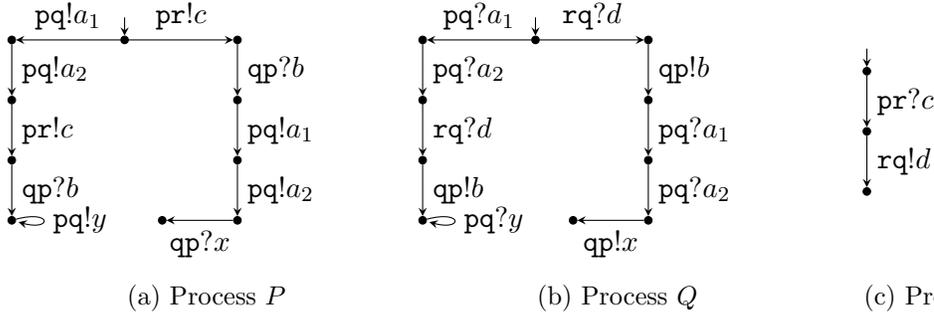
\begin{figure}
\centering
\begin{subfigure}[t]{0.35\textwidth}
\begin{tikzpicture}[>=stealth]
    \node[state] (q1) at (0cm,0cm) {};
    \node[state] (q2) at ([xshift=-1.5cm]q1) {};
    \node[state] (q3) at ([xshift=1.5cm]q1) {};
    \node[state] (q4) at ([yshift=-0.8cm]q2) {};
    \node[state] (q5) at ([yshift=-0.8cm]q4) {};
    \node[state] (q6) at ([yshift=-0.8cm]q5) {};
    \node[state] (q7) at ([yshift=-0.8cm]q3) {};
    \node[state] (q8) at ([yshift=-0.8cm]q7) {};
    \node[state] (q9) at ([yshift=-0.8cm]q8) {};
    \node[state] (q10) at ([xshift=-1cm]q9) {};

    \draw[->] ([yshift=0.3cm]q1.center) -- (q1);
    \draw[->] (q1) -- node[auto=right] {$\mathtt{pq}!a_1$} (q2);
    \draw[->] (q1) -- node[auto=left] {$\mathtt{pr}!c$} (q3);
    \draw[->] (q2) -- node[auto=left] {$\mathtt{pq}!a_2$} (q4);
    \draw[->] (q4) -- node[auto=left] {$\mathtt{pr}!c$} (q5);
    \draw[->] (q5) -- node[auto=left] {$\mathtt{qp}?b$} (q6);
    \draw[->] (q6) edge [loop right] node {$\mathtt{pq}!y$} (q6);
    \draw[->] (q3) -- node[auto=left] {$\mathtt{qp}?b$} (q7);
    \draw[->] (q7) -- node[auto=left] {$\mathtt{pq}!a_1$} (q8);
    \draw[->] (q8) -- node[auto=left] {$\mathtt{pq}!a_2$} (q9);
    \draw[->] (q9) -- node[auto=left] {$\mathtt{qp}?x$} (q10);
\end{tikzpicture}
\caption{Process $P$}
\end{subfigure}
\begin{subfigure}[t]{0.35\textwidth}
\begin{tikzpicture}[>=stealth]
    \node[state] (q1) at (0cm,0cm) {};
    \node[state] (q2) at ([xshift=-1.5cm]q1) {};
    \node[state] (q3) at ([xshift=1.5cm]q1) {};
    \node[state] (q4) at ([yshift=-0.8cm]q2) {};
    \node[state] (q5) at ([yshift=-0.8cm]q4) {};
    \node[state] (q6) at ([yshift=-0.8cm]q5) {};
    \node[state] (q7) at ([yshift=-0.8cm]q3) {};
    \node[state] (q8) at ([yshift=-0.8cm]q7) {};
    \node[state] (q9) at ([yshift=-0.8cm]q8) {};
    \node[state] (q10) at ([xshift=-1cm]q9) {};

    \draw[->] ([yshift=0.3cm]q1.center) -- (q1);
    \draw[->] (q1) -- node[auto=right] {$\mathtt{pq}?a_1$} (q2);
    \draw[->] (q1) -- node[auto=left] {$\mathtt{rq}?d$} (q3);
    \draw[->] (q2) -- node[auto=left] {$\mathtt{pq}?a_2$} (q4);
    \draw[->] (q4) -- node[auto=left] {$\mathtt{rq}?d$} (q5);
    \draw[->] (q5) -- node[auto=left] {$\mathtt{qp}!b$} (q6);
    \draw[->] (q6) edge [loop right] node {$\mathtt{pq}?y$} (q6);
    \draw[->] (q3) -- node[auto=left] {$\mathtt{qp}!b$} (q7);
    \draw[->] (q7) -- node[auto=left] {$\mathtt{pq}?a_1$} (q8);
    \draw[->] (q8) -- node[auto=left] {$\mathtt{pq}?a_2$} (q9);
    \draw[->] (q9) -- node[auto=left] {$\mathtt{qp}!x$} (q10);
\end{tikzpicture}
\caption{Process $Q$}
\end{subfigure}
\begin{subfigure}[t]{0.2\textwidth}
\begin{tikzpicture}[>=stealth]
    \node[state] (q1) at (0cm,0cm) {};
    \node[state] (q2) at ([yshift=-0.8cm]q1) {};
    \node[state] (q3) at ([yshift=-0.8cm]q2) {};

    \draw[->] ([yshift=0.3cm]q1.center) -- (q1);
    \draw[->] (q1) -- node[auto=left] {$\mathtt{pr}?c$} (q2);
    \draw[->] (q2) -- node[auto=left] {$\mathtt{rq}!d$} (q3);

    \path[use as bounding box] ([xshift=-0.5cm,yshift=0.5cm]q1) rectangle
    ([xshift=1cm,yshift=-1cm]q3);
\end{tikzpicture}
\caption{Process $R$}
\end{subfigure}
\caption{FIFO system of Example~\ref{ex:motivation} (from~\cite{DBLP:journals/corr/abs-1901-09606})}%
\label{fig:motivation}
\end{figure}

\subsection*{FIFO Machines}
\begin{defi}[FIFO machines]%
    \label{def:fifoSystems}
  A FIFO machine $S$ is a tuple $(Q, \channels, M, \Delta)$ where
  $Q$ is a finite set of control states,
  $\channels$ is a finite set of FIFO channels,
  $M$ is a finite message alphabet and $\Delta \subseteq (Q\times Q) \cup (Q \times
  (\channels \times \set{!,?} \times M) \times Q)$ is a finite set of
  transitions.
\end{defi}

We write a transition $(q, (\mathtt{c},?,a), q')$ as $q
\fstrans{\mathtt{c}?a} q'$; we similarly modify other transitions. We
call $q$ the source state and $q'$ the target state. Transitions
of the form $q \fstrans{\mathtt{c}?a} q'$ (resp.~$q
\fstrans{\mathtt{c}!a} q'$) denote
retrieve actions (resp.~send actions). Transitions of the form $q
\fstrans{} q'$ do not change the channel contents but only change the
control state.

The channels in $\channels$ hold strings in $M^{*}$. A channel
valuation $\fval$ is a fuction from $F$ to $M^*$. We denote the set of
all channel valuations by ${(M^{*})}^{\channels}$. Given two channel
valuations $\fval_{1}, \fval_{2} \in {(M^{*})}^{\channels}$, we denote
by $\fval_{1} \cdot \fval_{2}$ the valuation obtained by concatenating
the contents in $\fval_{1}$ and $\fval_{2}$ channel-wise. For a letter
$a \in M$ and a channel $\mathtt{c} \in \channels$, we denote by $\vec{a}_{\mathtt{c}}$
the channel valuation that assigns $a$ to $\mathtt{c}$ and $\epsilon$ to all
other channels.
The semantics of a FIFO machine $S$ is given by a transition system
$T_{S}$ whose set of states is $Q \times {(M^{*})}^{\channels}$, also
called configurations. Every
transition $q \fstrans{\mathtt{c}?a} q'$ of $S$ and channel valuation
$\fval \in {(M^{*})}^{\channels}$ results in the transition
$(q,\vec{a}_{\mathtt{c}}\cdot \fval) \fstrans{\mathtt{c}?a}
(q',\fval)$ in $T_{S}$. Every transition $q \fstrans{\mathtt{c}!a} q'$
of $S$ and channel valuation $\fval \in {(M^{*})}^{\channels}$ results
in the transition $(q,\fval) \fstrans{\mathtt{c}!a} (q',\fval\cdot
\vec{a}_{\mathtt{c}})$ in $T_{S}$. Intuitively, the transition
$q \fstrans{\mathtt{c}?a} q'$ (resp.~$q \fstrans{\mathtt{c}!a} q'$)
retrieves the letter $a$ from the front of the channel $\mathtt{c}$
(resp.~sends the letter $a$ to the back of the channel $\mathtt{c}$).
A run of $S$ is a (finite or infinite) sequence of configurations
$(q_{0}, \fval_{0}) (q_{1}, \fval_1) \cdots$ such that for every
$i \ge 0$, there is a transition $t_{i}$ such that $(q_{i},
\fval_{i}) \fstrans{t_{i}} (q_{i+1},\fval_{i+1})$.

\begin{exa}%
\label{ex:motivation}
Figure~\ref{fig:motivation} shows a FIFO system (from~\cite{DBLP:journals/corr/abs-1901-09606}) with three processes $P, Q,
R$ that communicate through four FIFO channels $pq, qp, pr, rq$.
Processes are FIFO machines where transitions are labeled
by sending or receiving operations with FIFO channels and, for
example, channel $pq$ is an unidirectional FIFO channel from process
$P$ to process $Q$. From this FIFO system, we get a FIFO
machine as given in Definition~\ref{def:fifoSystems} by product
construction.\footnote{We use \emph{FIFO machine} for one
automaton and \emph{FIFO system} when there are multiple
automata interacting with each other.} The control states of the
product FIFO machines are
triples, containing control states of processes $P, Q, R$. The product
FIFO machine can go from one control state to another if one of the
processes goes from a control state to another and the other two
processes remain in their states. For example, the product machine has
the transition $(q_{1},q_{2},q_{3}) \fstrans{\mathtt{pq}!a_{1}}
(q_{1}',q_{2}, q_{3})$, if process $P$ has the transition $q_{1}
\fstrans{\mathtt{pq}!a_{1}} q_{1}'$.
\end{exa}
%

For analyzing the running time of
algorithms, we assume the size of a machine to be the number of bits
needed to specify a machine (and source/target configurations if
necessary) using a reasonable encoding.
Let us begin to present the reachability problems that we tackle in
this paper.

\begin{problem}[Reachability]
\emph{Given:} A FIFO machine $S$ and two configurations
$(q_0,\fval_0)$ and $(q,\fval)$.
\emph{Question:} Is there a run starting from $(q_0,\fval_0)$ and
ending at $(q,\fval)$?
\end{problem}
\begin{problem}[Control-state reachability]
\emph{Given:}  A FIFO machine $S$, a configuration $(q_0,\fval_0)$ and a control-state $q$.
\emph{Question:} Is there a channel valuation $\fval$ such that $(q,\fval)$ is reachable from $(q_0,\fval_0)$?
\end{problem}
It is folklore that reachability and control-state reachability are
undecidable for machines operating on FIFO channels.

\subsection*{Flat machines} For a FIFO machine
$S=(Q,\channels,M,\Delta)$, its \emph{machine graph} $G_{S}$ is a
directed graph whose set of vertices is $Q$. There is a directed edge
from $q$ to $q'$ if there is some transition $q \fstrans{\mathtt{c}?a}
q'$ or $q \fstrans{\mathtt{c}!a} q'$ for some channel $\mathtt{c}$ and
some letter $a$, or there is a transition $q \fstrans{} q'$. We say
that $S$ is \emph{flat} if in $G_{S}$, every vertex is in at most one
directed cycle. Figure~\ref{fig:flatFIFO} shows a flat FIFO machine.
\begin{figure}[t]
    \begin{subfigure}[t]{0.45\textwidth}
        \centering
        \begin{tikzpicture}[>=stealth]
            \node[state, label=-90:$q_{0}$] (q0) at (0\ml,0\ml) {};
            \node[state, label=-90:$q_{1}$] (q1) at ([xshift=1cm]q0) {};
            \node[state, label=-90:$q_{2}$] (q2) at ([xshift=1cm]q1) {};
            \node[state, label=-90:$q_{3}$] (q3) at ([xshift=1cm]q2) {};
            \node[state] (q01) at (barycentric cs:q0=1,q1=1) {};
            \node[state] (q12) at (barycentric cs:q1=1,q2=1) {};
            \node[state] (q23) at (barycentric cs:q2=1,q3=1) {};

            \draw[->] (q0) -- (q01);
            \draw[->] (q01) -- (q1);
            \draw[->] (q1) -- (q12);
            \draw[->] (q12) -- (q2);
            \draw[->] (q2) -- (q23);
            \draw[->] (q23) -- (q3);

            \node[state, label=90:$q_{4}$] (q4) at ([yshift=1cm]q1) {};
            \node[state, label=90:$q_{5}$] (q5) at ([yshift=1cm]q2) {};
            \node[state] (q14) at ([xshift=-0.25cm]barycentric cs:q1=1,q4=1) {};
            \node[state] (q25) at ([xshift=-0.25cm]barycentric cs:q2=1,q5=1) {};
            \node[state] (p14) at ([xshift=0.25cm]barycentric cs:q1=1,q4=1) {};
            \node[state, label=0:$q_{6}$] (p25) at ([xshift=0.25cm]barycentric cs:q2=1,q5=1) {};

            \draw (q1) edge[->,bend left] (q14);
            \draw (q14) edge[->, bend left] (q4);
            \draw (q4) edge[->, bend left] (p14);
            \draw (p14) edge[->, bend left] (q1);
            \draw (q2) edge[->, bend left] (q25);
            \draw (q25) edge[->, bend left] (q5);
            \draw (q5) edge[->, bend left] (p25);
            \draw (p25) edge[->, bend left] (q2);
            \draw[->] (q1) edge[bend right=49] (q3);
            \draw[->] (q4) -- (q5);
            \draw[->] (p25) -- (q3);

            \node at (barycentric cs:q14=1,p14=1) {$\ell_{1}$};
            \node at (barycentric cs:q25=1,p25=1) {$\ell_{2}$};

            \node at ([yshift=0.2cm] barycentric cs:q0=1,q1=1) {};
            \node at ([yshift=0.2cm] barycentric cs:q1=1,q2=1) {};
            \node at ([yshift=0.2cm] barycentric cs:q2=1,q3=1) {};

            \node at ([xshift=0.03cm] barycentric cs:q1=1,q4=1) {};
            \node at ([xshift=0.03cm] barycentric cs:q2=1,q5=1) {};
        \end{tikzpicture}
        \caption{Flat FIFO machine}%
        \label{fig:flatFIFO}
    \end{subfigure}
    \begin{subfigure}[t]{0.5\textwidth}
        \centering
        \begin{tikzpicture}[>=stealth]
            \node[state, label=-90:$q_{0}$] (q0) at (0\ml,0\ml) {};
            \node[state, label=-90:$q_{1}$] (q1) at ([xshift=1cm]q0) {};
            \node[state, label=-90:$q_{2}$] (q2) at ([xshift=2cm]q1) {};
            \node[state, label=-90:$q_{3}$] (q3) at ([xshift=1cm]q2) {};

            \draw[longpath,->] (q0) -- (q1);
            \draw[longpath,->] (q1) -- (q2);
            \draw[longpath,->] (q2) -- (q3);

            \node[coordinate] (q4) at ([yshift=1cm]q1) {};
            \node[coordinate] (q5) at ([yshift=1cm]q2) {};

            \draw (q1) edge[longpath, bend left=45] (q4);
            \draw (q4) edge[longpath, bend left=45, ->] (q1);
            \draw (q2) edge[longpath, bend left=45] (q5);
            \draw (q5) edge[longpath, bend left=45, ->] (q2);

            \node at ([yshift=0.2cm] barycentric cs:q0=1,q1=1) {$p_0$};
            \node at ([yshift=0.2cm] barycentric cs:q1=1,q2=1) {$p_1$};
            \node at ([yshift=0.2cm] barycentric cs:q2=1,q3=1) {$p_2$};

            \node at ([xshift=0.03cm] barycentric cs:q1=1,q4=1) {$\ell_1$};
            \node at ([xshift=0.03cm] barycentric cs:q2=1,q5=1) {$\ell_2$};
        \end{tikzpicture}
        \caption{Path schema denoted by $p_0{(\ell_1)}^*p_1{(l_2)}^*p_2$}%
        \label{fig:pathSchema}
    \end{subfigure}
    \caption{Example flat FIFO machine and path schema}
\end{figure}
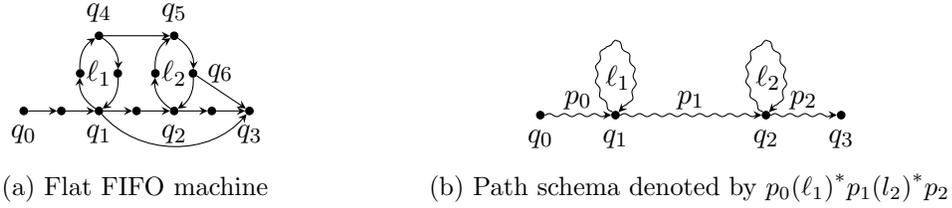

We call a FIFO machine $S=(Q,\channels,M,\Delta)$ a \emph{path segment} from
state $q_0$ to state $q_r$ if $Q=\set{q_0,\ldots,q_r}$,
$\Delta=\set{t_1,\ldots, t_r}$ and for every $i \in \set{1,\ldots,r}$,
$q_{i-1}$ is the source of $t_i$ and $q_i$ is its target. We call a
FIFO machine $S=(Q,\channels,M,\Delta)$ an \emph{elementary loop} on $q_0$ if $Q=\set{q_0,
\ldots, q_r}$, $\Delta=\set{t_1,\ldots, t_{r+1}}$ and for each $i\in
\set{1,\ldots,r+1}$, $t_i$ has source $q_{i-1}$ and target $q_{i \mod
(r+1)}$. We call $t_1 \cdots t_{r+1}$ the label of the loop.
A \emph{path schema} is a flat FIFO machine comprising of a sequence
$p_0\ell_1p_1\ell_2p_2\cdots l_{r}p_r$, where $p_0, \ldots, p_r$ are
path segments and $\ell_1, \ldots, \ell_r$ are elementary loops. There are
states $q_0, q_1, \ldots, q_{r+1}$ such that $p_0$ is a path segment
from $q_0$ to $q_1$ and for every $i \in \set{1, \ldots, r}$, $p_i$ is
a path segment from $q_i$ to $q_{i+1}$ and $\ell_i$ is an elementary loop on
$q_i$. Except $q_{i}$, none of the other states in $\ell_i$ appear in
other path segments or elementary loops. To emphasize that $\ell_1,
\ldots, \ell_r$ are elementary loops, we denote the path schema as
$p_0{(\ell_1)}^*p_1\cdots {(\ell_r)}^*p_r$. We use the term elementary
loop to distinguish them from loops in general, which may have some
states appearing more than once. All loops in flat FIFO machines are
elementary.  Figure~\ref{fig:pathSchema} shows a path schema, where
wavy lines indicate long path segments or elementary loops that may
have many intermediate states and transitions. This path schema is
obtained from the flat FIFO machine of Figure~\ref{fig:flatFIFO} by
removing the transitions from $q_{1}$ to $q_{3}$, $q_{4}$ to
$q_{5}$ and $q_{6}$ to $q_{3}$.

\begin{rem}[Fig.~\ref{fig:motivation}]
Each process $P, Q, R$ is flat and the cartesian product of the three
automata is almost flat except on one state:
there are two loops, one sending $y$ in channel $\mathtt{pq}$ and
another one retrieving $y$ from channel $\mathtt{pq}$.
\end{rem}

\subsection*{Notations and definitions}
 For any sequence $\sigma$ of transitions of a FIFO machine and channel
$\mathtt{c} \in \channels$, we denote by
$y^{\sigma}_{\mathtt{c}}$
(resp.~$x^{\sigma}_{\mathtt{c}}$) the sequence of
letters sent to (resp.~retrieved from) the channel $\mathtt{c}$ by
$\sigma$.
For a configuration $(q,\fval)$, let $\fval({\mathtt{c}})$ denote the
contents of channel $\mathtt{c}$.

\subsection*{Equations on words}

We recall some classical results reasoning about words and
prove one of them, to be used later.
The well-known Levi's Lemma says that the words $u,v \in \Sigma^*$
that are solutions of the equation $uv=vu$ satisfy $u,v \in z^*$ where
$z$ is a primitive word.
%
The solutions of the equation $uv=vw$ satisfy $u=xy,
w=yx, v={(xy)}^{n}x$, for some words $x,y$ and some integer $n\geq 0$.
The following lemma is used in~\cite{FPS-ICOMP} for exactly the same
purpose as here.

\begin{lem}%
    \label{lem:solutionInfWordEqn}
Consider three finite words $x,y \in \Sigma^+$ and $w \in \Sigma^*$. The equation $x^{\omega} =  wy^{\omega}$ holds iff
there exists a primitive word $z \neq \epsilon$ and two words $x',x''$ such
that $x=x'x''$, $x''x' \in z^*$, $w \in x^*x'$ and $y \in z^*$.
%
\end{lem}
\begin{proof}
Suppose $x,w,y$ satisfy the equation $x^{\omega} =  w.y^{\omega}$. If
$w=\epsilon$, then the equation reduces to $x^{\omega} =  y^{\omega}$.
Hence we deduce that $x^{\mid y \mid} =  y^{\mid x \mid}$. In this
case, we show (using Levi's Lemma and considering the three cases
$\mid x \mid =  \mid y \mid$ or $\mid x \mid <  \mid y \mid$ or $\mid
y \mid <  \mid x \mid$) that the solutions are the words $x,y \in z^*$
where $z$ is a finite primitive word.
Now suppose that $w \neq \epsilon$, so choose the smallest $n\geq0$ such that $w=x^{n}x'$ with $x=x'x''$. Hence, we obtain that ${(x''x')}^{\omega}=y^{\omega}$, and again we know that the solutions of this equation are $x''x',y \in z^*$  where $z$ is a primitive word.

For the converse, suppose $x=x'x''$, $x''x' = z^{j}$, $w = x^{n}x'$ and
$y = z^{k}$. We have $x^{\omega}=x^{n}x'{(x''x')}^{\omega}=w{(z^j)}^{\omega}=w{(z^k)}^{\omega}=wy^{\omega}$.
%
%
%
%
%
\end{proof}

%% file: complexity.tex
In this section, we give complexity bounds for the reachability
problem for flat FIFO machines. We also establish the complexity of
other related problems, viz.~repeated control state reachability,
termination, boundedness, channel boundedness and letter channel
boundedness. We use the algorithm for repeated control state
reachability as a subroutine for solving termination and boundedness.
For channel boundedness and letter channel boundedness, we use another
argument based on integer linear programming.
Flat FIFO machines can simulate counter machines and reachability and
related problems are known to be \NP{}-hard for flat counter machines.
However, the lower bound proofs for flat counter machines use binary
encoding of counter updates, while the simulation of counter machines
by FIFO machinbes use unary encoding. Hence, we cannot deduce lower bounds for flat
FIFO machines from the lower bounds for flat counter machines. We prove
the lower bounds for flat FIFO machines directly.

In~\cite{EGM2012}, Esparza,  Ganty, and Majumdar studied the
complexity of reachability for highly undecidable models
(multipushdown automata) but synchronized by bounded languages in the
context of bounded model-checking. In particular, they proved that
control-state reachability is \NP{}-complete for flat FIFO machines
(in fact for single-path FIFO machines, i.e. FIFO machines controlled
by a bounded language).
The \NP{}
upper bound is based on a simulation of FIFO path schemas by multi head
pushdown automata.
Some constraints need to be imposed on the multi head pushdown
automata to ensure the correctness of the simulation. The structure of
path schemas enables these constraints to be expressed as linear
constraints on integer variables and this leads to the \NP{} upper
bound.

Surprisingly, the \NP{} upper bound in~\cite{EGM2012} is given only
for the control-state reachability problem;
the complexity of the reachability problem is not
established in~\cite{EGM2012} while it is given for all other
considered models. However,
there is a simple linear reduction from reachability to control-state
reachability for FIFO (and Last In First Out) machines~\cite{sutre18}.
Such reductions are not known to exist for other models like counter
machines and vector addition systems.

We begin by reducing reachability to control-state reachability
(personal communication from Gr\'egoire Sutre~\cite{sutre18}) for
(general and flat) FIFO machines.

\begin{propC}[\cite{sutre18}]%
\label{reachability}
Reachability reduces (with a linear reduction) to control-state
reachability, for general FIFO machines and for flat FIFO machines.
\end{propC}

\begin{proof}
    Let $A$ be a  FIFO machine, $q$ a control-state and
    $(q,\fval)$ a configuration of $A$. We reduce reachability
    to control-state reachability. We construct the machine
    $B_{A,(q,\fval)}$ from $A$ and $(q,\fval)$ as follows. The machine
    $B_{A,(q,\fval)}$ is obtained from $A$ by adding a path to control
    state $q$ as follows, where $\#$ is a new symbol not in $M$
    and $\channels = \set{\mathtt{1}, \ldots, \mathtt{p}}$.
    The transition labeled $\mathtt{1} ? \fval(\mathtt{1})\#$ is to be understood
    as a sequence of transitions whose effect is to retrieve the string
    $\fval(\mathtt{1})\#$ from channel $\mathtt{1}$.
    \begin{center}
        \begin{tikzpicture}[>=stealth]
            \node[state, label=-90:$q$](q1) at (0cm,0cm) {};
            \node[state] (q2) at ([xshift=1.5cm]q1) {};
            \node[state] (q3) at ([xshift=1.75cm]q2) {};
            \node[state] (q4) at ([xshift=1cm]q3) {};
            \node[state] (q5) at ([xshift=1.5cm]q4) {};
            \node[state, label=-90:$q_{\mathrm{stop}}$] (q6) at ([xshift=1.75cm]q5) {};

            \draw[->] (q1) -- node[auto=left] {$\mathtt{1} ! \#$} (q2);
            \draw[->] (q2) -- node[auto=left] {$\mathtt{1} ? \fval(\mathtt{1})\#$} (q3);
            \draw[dotted,->] (q3) -- (q4);
            \draw[->] (q4) -- node[auto=left] {$\mathtt{p} ! \#$} (q5);
            \draw[->] (q5) -- node[auto=left] {$\mathtt{p} ? \fval(\mathtt{p})\#$} (q6);
        \end{tikzpicture}
    \end{center}


    The configuration $(q,\fval)$ is reachable in $A$ iff the control state
    $q_{\mathrm{stop}}$ is reachable in $B_{A,(q,\fval)}$. Note that
    if $A$ is flat, then $B_{A,(q,\fval)}$ is also flat.
\end{proof}

\begin{rem}
Control-state reachability is reducible to reachability for general
FIFO machines. Suppose $\Sigma = \set{a_{1}, \ldots, a_{d}}$ and there
are $\mathtt{p}$ channels. Using the same notations as in the previous proof,
from $A$ and $q$, one constructs the machine $B_{A,q}$ as follows: one
adds, to $A$, $d\times \mathtt{p}$ self loops $\ell_{i,j}$, each labeled by
$\mathtt{j}?a_i$, for $i \in \set{1,..,d}$ and $j \in \set{1, \ldots, \mathtt{p}}$, all
from and to the control-state $q$. We infer that $q$ is reachable in $A$
if and only if (by definition) there exists $\fval$ such that
$(q,\fval)$ is reachable in $A$ if and only if
$(q,\vec{\epsilon})$ is reachable in $B_{A,q}$. Here,
$(q,\vec{\epsilon})$ denotes the configuration where $q$ is the
control state and all channels are empty. Note that $B_{A,q}$ is not
necessarily flat, even if $A$ is flat. Hence, this reduction does not
imply \NP{}-hardness of reachability in flat FIFO machines. We will
prove \NP{}-hardness later using a different reduction.
\end{rem}

It is proved in~\cite[Theorem 7]{EGM2012} that control state
reachability is in \NP{} for flat FIFO machine.

\begin{cor}%
\label{cor:reachability}
Reachability is in \NP{} for flat FIFO machines.
\end{cor}

Now we define problems concerned with infinite behaviors.

\begin{problem}[Repeated reachability]
\emph{Given:} A FIFO machine $S$, two configurations $(q_0,\fval_0)$
and $(q,\fval)$. \emph{Question:} Is there an infinite run from $(q_0,\fval_0)$ such
that $(q,\fval)$ occurs infinitely often along this run?
\end{problem}
\begin{problem}[Cyclicity]
\emph{Given:} A FIFO machine $S$ and a configuration $(q,\fval)$.
\emph{Question:} Is $(q,\fval)$ reachable (by a non-empty run) from $(q,\fval)$?
\end{problem}
\begin{problem}[Repeated control-state reachability]
\emph{Given:} A FIFO machine $S$, a configuration $(q_0,\fval_0)$ and a
control-state $q$. \emph{Question:} Is there an infinite run from $(q_0,\fval_0)$ such
that $q$ occurs infinitely often along this run?
\end{problem}

We can easily obtain an  \NP{} upper bound for repeated reachability
in flat FIFO machines. A
non-deterministic Turing machine first uses the previous algorithm for
reachability (Corollary~\ref{cor:reachability}) to verify that $(q,\fval)$ is
reachable from $(q_0,\fval_0)$. Then the same algorithm is used again
to verify that $(q,\fval)$ is reachable from $(q,\fval)$ (i.e.
cyclic).

\begin{cor}
Repeated reachability is in  \NP{} for flat FIFO machines.
\end{cor}

Let us recall that the cyclicity property is \EXPSPACE-complete for
Petri nets~\cite{AF-ZB-INF-97,DBLP:journals/corr/DrewesL15} while
structural cyclicity (every configuration is cyclic) is in \PTIME\@.
Let us show that one may decide the cyclicity property for flat FIFO
machines in linear time.

\begin{lem}%
    \label{lem:cyclicity}
In a flat FIFO machine, a configuration $(q,\fval)$ is reachable from
$(q,\fval)$ iff there is an elementary loop labeled by
$\sigma$, such that $(q,\fval)  \fstrans{\sigma} (q,\fval)$.

\end{lem}
\begin{proof}
The implication from right to left ($\Leftarrow$) is clear. For the
converse, suppose that $(q,\fval)$ is reachable from $(q,\fval)$. Flatness
implies that $q$ belongs to a (necessarily unique and elementary) loop, say
a loop labeled by
$\sigma$. As $(q,\fval)$ is reachable from $(q,\fval)$, there exists a
sequence of transitions $\gamma$ such that $(q,\fval)
\fstrans{\gamma} (q,\fval)$. Now, still from flatness, $\gamma$ is
necessarily a power of $\sigma$, say $\gamma=\sigma^k$, $k \geq 1$.
Hence we have: $(q,\fval)  \fstrans{\sigma^k} (q,\fval)$. Let us write
$(q,\fval)  \fstrans{\sigma} (q,\fval_1)  \fstrans{\sigma} (q,\fval_2)
\fstrans{\sigma}\cdots \fstrans{\sigma} (q,\fval_k)=(q,\fval)$.
The effect of $\sigma$ on the channel contents
must preserve their initial length, so we have $\mid
x^{\sigma}_{\mathtt{c}} \mid =
\mid y^{\sigma}_{\mathtt{c}} \mid$ for every channel
$\mathtt{c}$.
Since $\sigma$ is fireable from $(q,\fval)$ and reaches $(q,\fval_1) $, let us show that $\fval_1 = \fval$.
If $x^{\sigma}_{\mathtt{c}} = \epsilon$ then
$x^{\sigma}_{\mathtt{c}}=y^{\sigma}_{\mathtt{c}}=\epsilon$ and $\fval_1 = \fval$.
%
So, let us suppose that
$x^{\sigma}_{\mathtt{c}} \neq \epsilon$ (this also implies $y^{\sigma}_{\mathtt{c}} \neq \epsilon$).
From $(q,\fval)  \fstrans{\sigma^k} (q,\fval)$, we know that the
sequence $\sigma^k$ is infinitely iterable and we have (1)
${({(x^{\sigma}_{\mathtt{c}})}^k)}^{\omega}={\fval_{\mathtt{c}}(
{(y^{\sigma}_{\mathtt{c}})}^k)}^{\omega}$ and since $k \geq 1$,
$x^{\sigma}_{\mathtt{c}} \neq \epsilon$ and $y^{\sigma}_{\mathtt{c}}
\neq \epsilon$, equality (1) implies that
${(x^{\sigma}_{\mathtt{c}})}^{\omega}=w{(y^{\sigma}_{\mathtt{c}})}^{\omega}$.
In the rest of this proof, we skip the superscript $\sigma$ and the
subscript $\mathtt{c}$ for simplicity. We now write $x^{\omega} = wy^{\omega}$.

Lemma~\ref{lem:solutionInfWordEqn} implies that there exists a
primitive word $z \neq \epsilon$ and two words $x',x''$ such
that $x=x'x''$, $x''x' \in z^*$, $w \in x^*x'$ and $y \in z^*$. Let us
write  $y=z^{d}$. Since $x''x' \in z^*$ and since $x''x'$ has the same
length as $y$, we deduce that $x''x' = z^{d} = y$. From $w \in
x^*x'$, we obtain that $w \in {(x'x'')}^*x'=x'{(x''x')}^*$, hence  $w \in x'{(z^{d})}^*$. Hence, we have:
\[y= x''x' = z^{d}, x=x'x''  \;\; \text{and}  \;\; w=x'z^{ds} \;\; \text{for some}  \;\; s \geq 0 \;\;  \;\; (2)  \]


Since $(q,\fval)  \fstrans{\sigma} (q,\fval_1)$, the firing equation
$w_1=x^{-1}wy$ is satisfied. By replacing $x,w$ by their values in
$(2)$ in the firing equation, we obtain:
\[w_1=x^{-1}wy=x^{-1}x'z^{ds}z^{d}=x''^{-1}z^{d}z^{ds}=x''^{-1}x''x'z^{ds}=x'z^{ds}=w.\]
Hence $(q,\fval)  \fstrans{\sigma} (q,\fval)$.
\end{proof}

To decide whether $(q,\fval)  \fstrans{*} (q,\fval)$, one tests
whether $(q,\fval)  \fstrans{\sigma} (q,\fval)$ for some elementary
loop $\sigma$ in the flat FIFO machine. Since the FIFO machine is flat,
$q$ can be in at most one loop, so only one loop need to be tested.
This gives
a linear time algorithm for deciding cyclicity.
%

\begin{cor}
Testing cyclicity can be done in linear time for flat FIFO machines.
\end{cor}

We are now going to show an NP upper bound for repeated control state reachability.

Let a loop be labeled with $\sigma$. Recall that for
each channel $\mathtt{c}$, we denote by $x^{\sigma}_{\mathtt{c}}$ (resp.~$y^{\sigma}_{\mathtt{c}}$) the projection of
$\sigma$ to letters retrieved from (resp.~sent to) the channel
$\mathtt{c}$.
Let us write $\sigma_\mathtt{c}$ for the projection of $\sigma$ on channel $\mathtt{c}$.

\begin{rem} The loop labeled by $\sigma$ is infinitely iterable from $(q,\fval)$ iff
    $\sigma_\mathtt{c}$ is infinitely iterable from
    $(q,\fval(\mathtt{c}))$, for every channel $\mathtt{c}$.
%
If $\sigma$ is infinitely iterable from $(q,\fval)$ then each
projection $\sigma_\mathtt{c}$ is also infinitely iterable from
$(q,\fval(\mathtt{c}))$. Conversely, suppose $\sigma_\mathtt{c}$ is
infinitely iterable from $(q,\fval(\mathtt{c}))$, for every channel
$\mathtt{c}$. For all $\mathtt{c} \neq \mathtt{c'}$, the actions of
$\sigma_\mathtt{c}$ and $\sigma_\mathtt{c'}$ are on different channels
and hence independent of each other. Since $\sigma$ is a shuffle of
$\set{\sigma_{\mathtt{c}} \mid \mathtt{c} \in F}$, we deduce that
$\sigma$ is infinitely iterable from $(q,\fval)$.
\end{rem}
%

We now give a characterization for a loop to be infinitely iterable.

\begin{lem}%
    \label{lem:characterizeInfIterLoops}
    Suppose an elementary loop is on a control state $q$ and is labeled by
    $\sigma$. It is infinitely iterable starting from the
    configuration $(q,\fval)$ iff for every channel $\mathtt{c}$,
    $x^{\sigma}_\mathtt{c}=\epsilon$
%
%
%
%
    or the following three conditions are true: $\sigma$ is fireable
    at least once from $(q,\fval)$,
    ${(x^{\sigma}_\mathtt{c})}^{\omega}=
    {\fval(\mathtt{c})\cdot(y^{\sigma}_\mathtt{c})}^{\omega}$ and
    $|x^{\sigma}_\mathtt{c}| \leq |y^{\sigma}_\mathtt{c}|$.
\end{lem}

\begin{proof}
Let $\ell$ be an elementary loop on a control state $q$ and labeled by
    $\sigma$.
If  $\sigma$ is infinitely iterable starting from the configuration
$(q,\fval)$ then for every channel $\mathtt{c}$, one has $|x_{\mathtt{c}}| \le
|y_{\mathtt{c}}|$. Otherwise, $|x_{\mathtt{c}}| > |y_{\mathtt{c}}|$ (the number of letters
retrieved is more than the number of letters sent in each
iteration), so the size of the channel content reduces with each
iteration, so there is a bound on the number of possible
iterations.
Since $\sigma$ is infinitely iterable from $(q,\fval)$,
the inequation ${(x^{\sigma}_\mathtt{c})}^n \leq
\fval(\mathtt{c})\cdot{(y^{\sigma}_\mathtt{c})}^n$
must hold for all $n\geq0$ (here, $\leq$ denotes the prefix relation).
If  $x_\mathtt{c} \neq \epsilon$, we may go at the limit and we obtain
${(x^{\sigma}_\mathtt{c})}^{\omega} \leq
\fval(\mathtt{c})\cdot{(y^{\sigma}_\mathtt{c})}^{\omega}$.


Finally, $\sigma$ is fireable at least once from $(q,\fval)$ since it is fireable infinitely from $(q,\fval)$.

Now conversely,  suppose that for every channel $\mathtt{c}$,
$x^{\sigma}_\mathtt{c} = \epsilon$ or the following three conditions
are true: $\sigma$ is fireable at least once from
$(q,\fval)$, ${(x^{\sigma}_\mathtt{c})}^{\omega}=
\fval(\mathtt{c})\cdot{(y^{\sigma}_\mathtt{c})}^{\omega}$ and
$|x^{\sigma}_\mathtt{c}| \leq |y^{\sigma}_\mathtt{c}|$.
%
%

For the rest of this proof, we fix a channel $\mathtt{c}$ and write
$x^{\sigma}_{\mathtt{c}}, y^{\sigma}_{\mathtt{c}},
\fval(\mathtt{c})$ as $x,y,w$ to simplify the notation.

If
$x = \epsilon$ then
$\sigma$ is infinitely iterable because it doesn't retrieve anything.
So assume that $x \neq \epsilon$.
%
%
%
We have $x^{\omega} = wy^\omega$ from the hypothesis. We infer from
Lemma~\ref{lem:solutionInfWordEqn} that there is a primitive word
$z \ne \epsilon$ and words $x',x''$ such that $x=x'x''$, $x''x' \in
z^{*}$, $w \in x^{*}x'$ and $y \in z^{*}$. Suppose $x''x'=z^{j}$ and
$y = z^{k}$. Since $|y| \ge |x| = |x''x'|$, we have $k \ge j$. Let us
prove the following monotonicity property: for all $n \geq 0$, $\sigma$ is
fireable from any channel content $w z^n$ and the resulting channel
content is $w
z^{n+(k-j)}$ (this will imply that for all $m\geq 1$,
$w \fstrans{\sigma^m} w z^{m \times (k-j)}$, hence that $\sigma$ is
infinitely iterable). We prove the monotonicity property by induction on
$n$.

For the base case $n=0$, we need to prove that $w  \fstrans{\sigma} w
z^{k-j}$. By hypothesis, $\sigma$ is fireable at least once from
$w$, hence
$w  \fstrans{\sigma} w'$ for some $w'$. We have $w'=x^{-1}w y = x^{-1}
x^{r}x'z^{k}$ for some $r \in \Nat$. Since $k \geq j$, we have $w'= x^{-1}
x^{r}x'z^{j}z^{k-j}=x^{-1} x^{r}x'(x''x')z^{k-j}=x^{-1}
x^r(x'x'')x'z^{k-j}=x^{-1} x^{r+1}x'z^{k-j}=x^{r}x'z^{k-j}=w z^{k-j}$.
%
%

For the induction step, we have to show that $\sigma$ is fireable from
channel content $w z^{n+1}$ and the resulting channel content is $w z^{n+1+(k-j)}$.
From induction hypothesis, we know that $\sigma$ is fireable from
channel content $w z^n$. Since $y=z^{k}$, the channel content after
firing a prefix $\sigma_{1}$ of $\sigma$ is
$x_{1}^{-1}wz^{n}z^{s}z_{1}$, where $x_{1}$ is some prefix of
$x$, $s \in \Nat$ and $z_{1}$ is some prefix of $z$. By induction on
$|\sigma_{1}|$, we can
verify that $\sigma_{1}$ can be fired from $wz^{n+1}$ and results in
$x_{1}^{-1}wz^{n+1}z^{s}z_{1}$. Hence, $\sigma$ can be fired from
$wz^{n+1}$ and results in $x^{-1}wz^{n+1}y =
x^{-1}x^{r}x'z^{n+1}z^{k} =
x^{-1}x^{r}x'z^{j}z^{n+1+k-j}=x^{-1}x^{r}x'x''x'z^{n+1+k-j}=x^{-1}x^{r+1}x'z^{n+1+k-j}=wz^{n+1+k-j}$.
This completes the induction step and hence proves the monotonicity
property.

Hence $\sigma$ is infinitely
iterable.
\end{proof}

The proof of Lemma~\ref{lem:characterizeInfIterLoops} provides a
complete characterization of the contents of a FIFO channel when a
loop is infinitely iterable. One may observe that the channel acts
like a counter (of the number of  occurrences of $z$).

\begin{cor}
    With the previous notations, the set of words in channel
    $\mathtt{c}$ that
occur in control-state $q$ is the regular periodic language
$\fval(\mathtt{c})\cdot{[z_\mathtt{c}^{k-j}]}^*$, when the elementary
loop containing $q$ is iterated arbitrarily many times.
\end{cor}

\begin{rem}
One may find other similar results on infinitely iterable loops in
many papers~\cite{finkel86,DBLP:journals/tcs/JeronJ93,DBLP:conf/sas/BoigelotGWW97,DBLP:journals/tcs/BouajjaniH99,FPS-ICOMP}.
Our Lemma~\ref{lem:characterizeInfIterLoops} is the same as~\cite[Proposition 5.1]{FPS-ICOMP} except that it (easily) extends
it to machines with \emph{multiple} channels and also provides the converse.
Lemma~\ref{lem:characterizeInfIterLoops} simplifies and improves
Proposition~5.4 in~\cite{DBLP:journals/tcs/BouajjaniH99} that used
the equivalent but more complex notion of \emph{inc-repeating
sequence}. Also, the results in~\cite{DBLP:journals/tcs/BouajjaniH99} don't give the simple representation of the regular periodic language.
\end{rem}

\begin{prop}%
    \label{lem:RCSRInNP}
    The repeated control state reachability problem is in
    \NP{} for flat FIFO machines.
\end{prop}
\begin{proof}
    We describe an \NP{} algorithm. Suppose $S$ is the given flat FIFO
    machine and the control state
    $q$ is to be reached repeatedly. Suppose $q$ is in a loop labeled
    with $\sigma$. The algorithm first verifies that for every
    channel $\mathtt{c}$, $|x^{\sigma}_{\mathtt{c}}| \le |y^{\sigma}_{\mathtt{c}}|$ --- if this condition is
    violated, the answer is no. From
    Lemma~\ref{lem:characterizeInfIterLoops}, it is enough to verify
    that we can reach a configuration $(q,\fval)$ such that
    $\sigma$ can be fired at least once from
    $(q,\fval)$ and for
    every channel $\mathtt{c}$ for which $x^{\sigma}_{\mathtt{c}} \ne \epsilon$, we have
    ${(x^{\sigma}_{\mathtt{c}})}^{\omega} =
    \fval({\mathtt{c}})\cdot{(y^{\sigma}_{\mathtt{c}})}^{\omega}$. Since the case
    of $x^{\sigma}_{\mathtt{c}} = \epsilon$ can be handled easily, we assume in the rest
    of this proof that $x^{\sigma}_{\mathtt{c}} \ne \epsilon$ for every $\mathtt{c}$. For verifying that
    ${(x^{\sigma}_{\mathtt{c}})}^{\omega} =
    \fval({\mathtt{c}})\cdot{(y^{\sigma}_{\mathtt{c}})}^{\omega}$, the algorithm
    depends on Lemma~\ref{lem:solutionInfWordEqn}: the algorithm
    guesses $x'_{\mathtt{c}}, x''_{\mathtt{c}}, z_{\mathtt{c}} \in M^{*}$ such that
    $x^{\sigma}_{\mathtt{c}} = x'_{\mathtt{c}}x''_{\mathtt{c}}$ and
    $x''_{\mathtt{c}}x'_{\mathtt{c}},y^{\sigma}_{\mathtt{c}} \in
    z_{\mathtt{c}}^{*}$.
    We have $|x'_{\mathtt{c}}|, |x''_{\mathtt{c}}| \le |x^{\sigma}_{\mathtt{c}}|$ and $|z_{\mathtt{c}}| \le
    |y^{\sigma}_{\mathtt{c}}|$ so the guessed strings are of size bounded by the size of
    the input. It remains to verify that we can reach a configuration
    $(q,\fval)$ such that for every channel $\mathtt{c}$,
    $\fval({\mathtt{c}}) \in {(x^{\sigma}_{\mathtt{c}})}^{*}x'_{\mathtt{c}}$ and $\sigma$ can be fired at
    least once from $(q,\fval)$. For
    accomplishing these two tasks, we add
    a channel $\mathtt{c}'$ for every channel $\mathtt{c}$ in the FIFO
machine $S$.
    The following gadgets are appended to the control state $q$, assuming
    that there are $\mathtt{p}$ channels and $\#$ is a special letter not in
    the channel alphabet $M$. We denote by $\sigma'$ the sequence of
    transitions obtained from $\sigma$ by replacing every channel
    $\mathtt{c}$ by $\mathtt{c}'$. A transition labeled with
    $\mathtt{c}?x^{\sigma}_{\mathtt{c}}; \mathtt{c}'!x^{\sigma}_{\mathtt{c}}$ is to
    be understood as a sequence of transitions whose effect is to
    retrieve
    $x^{\sigma}_{\mathtt{c}}$ from channel $\mathtt{c}$ and send $x^{\sigma}_{\mathtt{c}}$ to channel $\mathtt{c}'$.

    \begin{center}
        \begin{tikzpicture}[>=stealth]
            \node[state, label=90:$q$] (q) at (0cm,0cm) {};
            \node[state] (q1) at ([xshift=1cm]q) {};
            \node[state] (q2) at ([xshift=1.8cm]q1) {};
            \node[state] (q3) at ([xshift=1cm]q2) {};
            \node[state] (q4) at ([xshift=1cm]q3) {};
            \node[state] (q5) at ([xshift=2cm]q4) {};
            \node[state] (q6) at ([xshift=1cm]q5) {};
            \node[state] (q7) at ([xshift=0.7cm]q6) {};
            \node[state] (q8) at ([xshift=1cm]q7) {};
            \node[state] (q9) at ([xshift=2cm]q8) {};
            \node[state, label=90:$q'$] (q10) at ([xshift=1cm]q9) {};
            \node[state, label=90:$q_{f}$] (q11) at ([xshift=1cm]q10) {};

            \draw[->] (q) -- node[auto=right] {$\mathtt{1} ! \#$} (q1);
            \path[->] (q1) edge [loop above] node {$\mathtt{1} ? x^{\sigma}_{\mathtt{1}}; \mathtt{1}' !  x^{\sigma}_{\mathtt{1}}$} (q1);
            \draw[->] (q1) -- node[auto=right] {$\mathtt{1} ? x'_{\mathtt{1}}; \mathtt{1}' ! x'_{\mathtt{1}}$} (q2);
            \draw[->] (q2) -- node[auto=right] {$\mathtt{1} ? \#$} (q3);

            \draw[->] (q3) -- node[auto=right] {$\mathtt{2} ! \#$} (q4);
            \path[->] (q4) edge [loop above] node {$\mathtt{2} ? x^{\sigma}_{\mathtt{2}}; \mathtt{2}'!x^{\sigma}_{\mathtt{2}}$} (q4);
            \draw[->] (q4) -- node[auto=right] {$\mathtt{2} ? x'_{\mathtt{2}}; \mathtt{2}'!x'_{\mathtt{2}}$} (q5);
            \draw[->] (q5) -- node[auto=right] {$\mathtt{2} ? \#$} (q6);

            \draw[->, dotted] (q6) -- (q7);

            \draw[->] (q7) -- node[auto=right] {$\mathtt{p} ! \#$} (q8);
            \path[->] (q8) edge [loop above] node {$\mathtt{p} ? x^{\sigma}_{\mathtt{p}};\mathtt{p}'!x^{\sigma}_{\mathtt{p}}$} (q8);
            \draw[->] (q8) -- node[auto=right] {$\mathtt{p} ? x'_{\mathtt{p}};\mathtt{p}'!x'_{\mathtt{p}}$} (q9);
            \draw[->] (q9) -- node[auto=right] {$\mathtt{p} ? \#$} (q10);
            \draw[->] (q10) -- node[auto=right] {$\sigma'$} (q11);
        \end{tikzpicture}
    \end{center}

    Finally our algorithm runs the \NP{} algorithm to check that the control state
    $q_{f}$ is reachable. We claim that the control state $q$ can be visited infinitely
    often iff our algorithm accepts. Suppose $q$ can be visited
    infinitely often. So the loop containing $q$ can be iterated
    infinitely often. Hence from Lemma~\ref{lem:characterizeInfIterLoops},
    we infer that $S$ can reach a configuration $(q,
    \fval)$ such that $\sigma$ can be fired at least once and for every channel $\mathtt{c}$,
    $|x^{\sigma}_{\mathtt{c}}| \le |y^{\sigma}_{\mathtt{c}}|$ and
    ${(x^{\sigma}_{\mathtt{c}})}^{\omega} =
    \fval({\mathtt{c}})\cdot{(y^{\sigma}_{\mathtt{c}})}^{\omega}$. From
    Lemma~\ref{lem:solutionInfWordEqn}, there
    exist $x'_{\mathtt{c}}, x''_{\mathtt{c}}, z_{\mathtt{c}} \in M^{*}$ such that $x^{\sigma}_{\mathtt{c}} =
    x'_{\mathtt{c}}x''_{\mathtt{c}}$, $\fval({\mathtt{c}}) \in
    {(x^{\sigma}_{\mathtt{c}})}^{*}x'_{\mathtt{c}}$ and
    $x''_{\mathtt{c}}x'_{\mathtt{c}},y^{\sigma}_{\mathtt{c}} \in z_{\mathtt{c}}^{*}$. Our algorithm can guess
    exactly these words $x'_{\mathtt{c}}, x''_{\mathtt{c}},z_{\mathtt{c}}$. It is easy to verify
    that from the
    configuration $(q,\fval)$, the configuration $(q',
    \fval')$ can
    be reached, where $\fval'({\mathtt{c}'}) = \fval({\mathtt{c}})$ for every
    $\mathtt{c}$. Since
    $\sigma$ can be fired from $(q,\fval)$, $\sigma'$ can be
    fired from $(q',
    \fval')$ to reach $q_{f}$. So our algorithm accepts.

    Conversely, suppose our
    algorithm accepts. Hence the control state $q_{f}$ is reachable.
    By construction, we can verify that the run reaching the control
    state $q_{f}$ has to visit a configuration $(q,\fval)$ such
    that for every channel $\mathtt{c}$, $\fval({\mathtt{c}}) \in
    {(x^{\sigma}_{\mathtt{c}})}^{*}x'_{\mathtt{c}}$
    and $\sigma$ can be fired at least once from $(q,\fval)$.
    Our algorithm also verifies that $|x^{\sigma}_{\mathtt{c}}| \le |y^{\sigma}_{\mathtt{c}}|$, $x^{\sigma}_{\mathtt{c}} =
    x'_{\mathtt{c}}x''_{\mathtt{c}}$ and $x''_{\mathtt{c}}x'_{\mathtt{c}},y^{\sigma}_{\mathtt{c}} \in z_{\mathtt{c}}^{*}$. Hence,
    from Lemma~\ref{lem:solutionInfWordEqn} and
    Lemma~\ref{lem:characterizeInfIterLoops}, we infer that the loop
    containing $q$ can be iterated infinitely often starting from the
    configuration $(q,\fval)$. Hence, there is a run that
    visits $q$ infinitely often.
\end{proof}

Let us now introduce the non-termination and the unboundedness problems.

\begin{problem}[Non-termination]
\emph{Given:} A FIFO machine $S$ and an initial configuration
$(q_0,\fval_0)$. \emph{Question:} Is there an infinite run from
$(q_0,\fval_0)$?
\end{problem}
\begin{problem}[Unboundedness]
\emph{Given:} A FIFO machine $S$ and an initial configuration
$(q_0,\fval_0)$. \emph{Question:} Is the set of configurations
reachable from $(q_0,\fval_0)$
infinite?
\end{problem}

\begin{cor}
\label{cor:nonTerminationUnboundedness}
For flat FIFO machines, the
non-termination and unboundedness problems are in \NP{}.
\end{cor}
\begin{proof}
First we deal with non-termination. A flat machine is non-terminating
iff there is an infinite run $r$. As there are only a finite number of
control-states, the run will visit at least one control state (say $q$)
infinitely often. Hence to solve non-termination, we can guess a
control state $q$ and use the \NP{} algorithm of
Proposition~\ref{lem:RCSRInNP} to check that $q$ can be visited infinitely
often. This gives an \NP{} upper bound for non-termination.

Next we deal with unboundedness. The \emph{effect of a loop $\ell$}
labeled with $\sigma$ is a vector of integers
$\vec{v}_{\ell} \in \Int^\channels$ such that
$\vec{v}_{\ell}(\mathtt{c})=|x^{\sigma}_\mathtt{c}| - |y^{\sigma}_\mathtt{c}
|$ for every $\mathtt{c}\in \channels$. If $\ell$ is an infinitely iterable loop, then $\vec{v}_{\ell}\geq
\vec{0}$, where $\geq$ is component-wise comparison and $\vec{0}$ is
the vector with all components equal to $0$.
If none of the loops in a flat FIFO
machine are infinitely iterable, then only finitely many configurations
can be reached. Hence, an unbounded flat FIFO machine has at least one
loop $\ell$ that is infinitely iterable, hence $\vec{v}_{\ell} \ge
\vec{0}$. If every infinitely iterable loop $\ell$ has
$\vec{v}_{\ell}=\vec{0}$, then none of the infinitely iterable loops
will increase the length of any channel content. Hence, there is a
bound on the length of the channel contents in any reachable
configuration, so only finitely many configurations can be reached.
Hence, in an unbounded flat FIFO machine, there is at least one
infinitely iterable loop $\ell$ with $\vec{v}_{\ell}\ne 0$.

Conversely, suppose a flat FIFO machine has an infinitely iterable loop
$\ell$ with $\vec{v}_{\ell}\ne \vec{0}$. Since $\ell$ is infinitely
iterable, $\vec{v}_{\ell}\ge \vec{0}$. Hence there is some channel
$\mathtt{c}$ such that $\vec{v}_{\ell}(\mathtt{c}) \ge 1$. So every
iteration of the loop $\ell$ will increase the length of the content
of channel $\mathtt{c}$ by at least $1$. Hence, infinitely many
iterations of the loop $\ell$ will result in infinitely many
configurations.  So a machine $S$ is unbounded iff there exists an
infinitely iterable loop $\ell$ such that $\vec{v}_{\ell} \ge \vec{0}$
and $\vec{v}_{\ell} \ne \vec{0}$. Hence to decide unboundedness, we
guess a control state $q$, verify that it belongs to a loop whose
effect is non-negative on all channels and strictly positive on at
least one channel and use the algorithm of Proposition~\ref{lem:RCSRInNP} to
check that $q$ can be visited infinitely often. This gives an \NP{}
upper bound for unboundedness.
\end{proof}

For a word $w$ and a letter $a$, $|w|_{a}$ denotes the number of
occurrences of $a$ in $w$. For a FIFO machine, we say that a letter $a$
is unbounded in channel $\mathtt{c}$ if for every number $B$, there
exists a reachable configuration $(q,\fval)$ with
$|\fval(\mathtt{c})|_a \geq B$.
A channel $\mathtt{c}$ is unbounded if at least one letter $a$ is
unbounded in $\mathtt{c}$.

\begin{problem}[Channel-unboundedness]
    \emph{Given:} A FIFO machine $S$, an initial configuration
    $(q_0,\fval_0)$ and a channel $\mathtt{c}$. \emph{Question:} Is
    the channel $\mathtt{c}$ unbounded from $(q_0,\fval_0)$?
\end{problem}
\begin{problem}[Letter-channel-unboundedness]
    \emph{Given:} A FIFO machine $S$, an initial configuration
$(q_0,\fval_0)$, a channel $\mathtt{c}$ and a letter $a$.
\emph{Question:} Is the letter $a$ unbounded in channel $\mathtt{c}$
from $(q_0,\fval_0)$?
\end{problem}

Now we give an \NP{} upper bound for letter channel unboundedness in
flat FIFO machines. We
use the following two results in our proof.
\begin{thmC}[{\cite[Theorem 3, Theorem 7]{EGM2012}}]%
    \label{thm:fifoPathSchemaPresburgerFormula}
    Let $S = p_{0}{(\ell_{1})}^{*}p_{1} \cdots
    {(\ell_{r})}^{*}p_{r}$ be a FIFO path schema. We can compute in
    polynomial time an existential Presburger formula
    $\phi(x_{1}, \ldots, x_{r})$ satisfying the following property:
    there is a run of $S$ in which the loop $\ell_{i}$ is iterated
    exactly $n_{i}$ times for every $i \in \set{1, \ldots, r}$ iff
    $\phi(n_{1}, \ldots, n_{r})$ is true.
\end{thmC}
    For vectors $\vec{k},\vec{x}$ and matrix $\vec{A}$, the expression $\vec{k} \cdot
    \vec{x}$ denotes the dot product and the expression $\vec{A}\vec{x}$ denotes the
    matrix product.
    \begin{lemC}[{\cite[Lemma 3]{Papadimitriou1981}}]%
    \label{thm:maximiseIntegerLinearProgram}
  Suppose $\vec{A}$ is an integer matrix and
    $\vec{k},\vec{b}$ are integer vectors satisfying the following
    property: for every $B \in \Nat$, there exists a vector $\vec{x}$
    of \emph{rational numbers} such that $\vec{A}\vec{x}\ge \vec{b}$
    and $\vec{k}\cdot\vec{x} \ge B$. If there is an integer vector $\vec{x}$
    such that $\vec{A}\vec{x} \ge \vec{b}$, then for every $B \in
    \Nat$, there exists an \emph{integer} vector $\vec{x}$ such that
    $\vec{A}\vec{x} \ge \vec{b}$ and $\vec{k}\cdot\vec{x} \ge B$.
\end{lemC}

\begin{prop}\label{channel-boundedness}
Given a flat FIFO machine, a letter $a$ and channel $\mathtt{c}$, the problem of
checking whether $a$ is unbounded in $\mathtt{c}$ is in \NP{}.
\end{prop}
\begin{proof}
The letter $a$ is unbounded in $\mathtt{c}$ iff there exists a control state
$q$ such that for every number $B$, there is a reachable configuration
with control state $q$ and at least $B$ occurrences of $a$ in
channel $\mathtt{c}$ (this follows from definitions since there are
only finitely many control states). A non-deterministic polynomial
time Turing machine begins by guessing a control state $q$. If there are $r$ loops in
the path schema ending at $q$, the Turing machine computes an
existential Presburger formula $\phi(x_1, \ldots, x_r)$ satisfying the
following property: $\phi(n_1, \ldots, n_r)$ is true iff there is a
run ending at $q$ in which loop $i$ is iterated $n_i$ times for every
$i \in \set{1, \ldots, r}$. Such a formula can be computed in
polynomial time (Theorem~\ref{thm:fifoPathSchemaPresburgerFormula}).
 Let
$k_i$ be the number of occurrences of the letter $a$ sent to channel $\mathtt{c}$ by one iteration of the
$i$\textsuperscript{th} loop ($k_i$ would be negative if $a$ is
retrieved instead). If loop $i$ is iterated $n_i$ times for
every $i$ in a run, then at the end of the run there are $k_1 n_1 + \cdots +
k_r n_r$ occurrences of the letter $a$ in channel $\mathtt{c}$. To check that $a$
is unbounded in channel $\mathtt{c}$, we have to verify that there are tuples
$\langle n_1, \ldots, n_r \rangle$ such that $\phi(n_1, \ldots, n_r)$
is true and  $k_1 n_1 + \cdots + k_r n_r$ is arbitrarily large. This
is easier to do if there are no disjunctions in the formula $\phi(x_1,
\ldots, x_r)$. If there are any sub-formulas with disjunctions, the
Turing machine non-deterministically chooses one of the disjuncts and
drops the other one. This is continued till all disjuncts are
discarded. This results in a conjunction of linear inequalities, say
$A\vec{x} \ge \vec{b}$, where $\vec{x}$ is the tuple of variables
$\langle x_1, \ldots, x_r \rangle$. The machine then tries to maximize $k_1 x_1 +
\cdots + k_r x_r$ over rationals subject to the constraints $A \vec{x}
\ge \vec{b}$. This can be done in polynomial time, since linear
programming is in polynomial time. If the value $k_1 x_1 + \cdots + k_r x_r$
is unbounded above
over rationals subject to the constraints $A \vec{x} \ge \vec{b}$, then the machine invokes the \NP{} algorithm to check
if the constraints $A \vec{x} \ge \vec{b}$ has a feasible solution
over integers. If it does, then $k_1 x_1 + \cdots + k_r x_r$ is
also unbounded above over integers
(Lemma~\ref{thm:maximiseIntegerLinearProgram}). Hence, in this case, $a$ is unbounded in
channel $\mathtt{c}$.
\end{proof}

The above result also gives an \NP{} upper bound for
channel-unboundedness. We just guess a letter $a$ and check that it is
unbounded in the given channel.

We adapt the proof of \NP{}-hardness for the control state
reachability problem from~\cite{EGM2012} to prove \NP{} hardness for
reachability, repeated control state reachability, unboundedness and
non-termination.
\begin{thm}%
    \label{lem:reachNPHard}
    For flat FIFO machines, reachability,
    repeated control-state reachability, non-termination, unboundedness,
    channel-unboundedness and letter-channel-unboundedness are
    NP-hard.
\end{thm}
\begin{proof}
    We reduce from 3SAT\@. Given a 3-CNF formula
    $\mathrm{clause}_{1} \land \cdots \land \mathrm{clause}_{m}$ over
    variables $x_{1}, \ldots, x_{n}$, we construct a flat FIFO machine
with $n$ channels $\set{x_{1}, \ldots,x_{n}}$. There are two letters $0,1$ in the
    message alphabet. The channel $x_{i}$ is used to keep a guess of
    the truth assignment to the variable $x_{i}$. The flat FIFO
machine consists
    of the gadgets shown in Fig.~\ref{fig:reachNPHard}.
    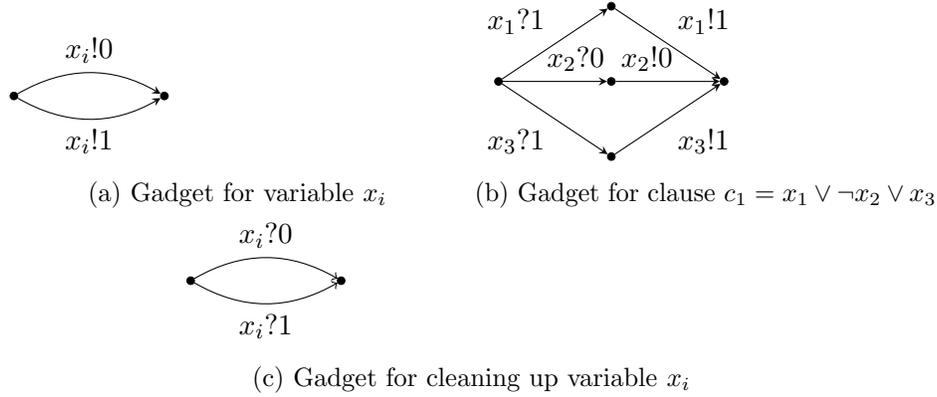
\begin{figure}[!htp]
        \centering
            \begin{subfigure}[t]{0.4\textwidth}
                \begin{tikzpicture}[>=stealth]
                    \node[state] (q0) at (0cm,0cm) {};
                    \node[state] (q1) at ([xshift=2cm]q0) {};

                    \draw[->] (q0) edge[bend left] node[auto=left] {$x_{i} !0$} (q1);
                    \draw[->] (q0) edge[bend right] node[auto=right] {$x_{i} !1$} (q1);
                \end{tikzpicture}
                \caption{Gadget for variable $x_{i}$}
            \end{subfigure}
            \begin{subfigure}[t]{0.4\textwidth}
                \begin{tikzpicture}[>=stealth]
                    \node[state] (q0) at (0cm,0cm) {};
                    \node[state] (q1) at ([xshift=3cm]q0) {};

                    \node[state] (qt) at ([yshift=1cm]barycentric cs:q0=1,q1=1) {};
                    \node[state] (qm) at (barycentric cs:q0=1,q1=1) {};
                    \node[state] (qb) at ([yshift=-1cm]barycentric cs:q0=1,q1=1) {};

                    \draw[->] (q0) -- node[auto=left] {$x_{1}?1$} (qt);
                    \draw[->] (qt) -- node[auto=left] {$x_{1}!1$} (q1);
                    \draw[->] (q0) -- node[auto=left,pos=0.7] {$x_{2}?0$} (qm);
                    \draw[->] (qm) -- node[auto=left,pos=0.3] {$x_{2}!0$} (q1);
                    \draw[->] (q0) -- node[auto=right] {$x_{3}?1$} (qb);
                    \draw[->] (qb) -- node[auto=right] {$x_{3}!1$} (q1);
                \end{tikzpicture}
                \caption{Gadget for clause $c_{1} = x_{1} \lor \lnot
                x_{2} \lor x_{3}$}
            \end{subfigure}
            \begin{subfigure}[t]{0.5\textwidth}
            \begin{tikzpicture}
                    \node[state] (q0) at (0cm,0cm) {};
                    \node[state] (q1) at ([xshift=2cm]q0) {};

                    \draw[->] (q0) edge[bend left] node[auto=left] {$x_{i} ?0$} (q1);
                    \draw[->] (q0) edge[bend right] node[auto=right] {$x_{i} ?1$} (q1);
            \end{tikzpicture}
            \caption{Gadget for cleaning up variable $x_{i}$}
            \end{subfigure}
            \caption{Gadgets used in the proof of Theorem~\ref{lem:reachNPHard}}%
        \label{fig:reachNPHard}
    \end{figure}
    The gadget for variable $x_{i}$ adds either $0$ (in the top transition)
    or $1$ (in the bottom edge) to channel $x_{i}$. At the end of this
    gadget, channel $x_{i}$ will have either $0$ or $1$. We will sequentially compose the
    gadgets for all variables. Starting from the initial control state
    of the gadget for variable $x_{1}$, we reach the final control
    state of the gadget for variable $x_{n}$ and the contents of the
    channels $x_{1}, \ldots, x_{n}$ determine a truth valuation.

    The gadget for the example clause $c_{1} =
    x_{1} \lor \lnot x_{2} \lor x_{3}$ (gadgets for other clauses
    follow similar pattern) is shown in Fig.~\ref{fig:reachNPHard}.
    The gadget checks that channel $x_{1}$ has $1$ (in the top path)
    or that channel $x_{2}$ has $0$ (in the middle path) or that
    channel $3$ has $1$ (in the bottom path). We
    append the clause gadgets to the end of the variable gadgets one
    after the other. All clauses are satisfied by the truth valuation
    determined by the contents of channels $x_{1}, \ldots, x_{n}$ iff
    we can reach the last control state of the last clause.

    The gadget for cleaning up variable $x_{i}$ is shown on the bottom
    in Fig.~\ref{fig:reachNPHard}. We append the cleanup gadgets to
    the end of the clause gadgets one after the other.

    The given 3-CNF formula is satisfiable iff the last control state
    of the cleanup gadget for variable $x_{n}$ can be reached with all
    channels being empty. Hence, this constitutes a reduction to the
    reachability problem. Note that in the flat FIFO machine
    constructed above, there are no loops, so all channels are bounded and none of the
    control states can be visited infinitely often. We add a self loop
    to the last control state of the cleanup gadget for variable
    $x_{n}$ that adds letter $1$ to channel $x_{1}$. If this loop can
    be reached, then it can be iterated infinitely often to add
    unboundedly many occurrences of the letter $1$ to channel
    $x_{1}$. Now, the given 3-CNF formula is satisfiable iff the
    constructed flat FIFO machine is unbounded iff channel $x_{1}$ is
    unbounded iff letter $1$ is unbounded in channel $x_{1}$ iff there
    is a non-terminating run iff the last control state of the cleanup
    gadget for variable $x_{n}$ can be reached infinitely often. Hence
    reachability, unboundedness, channel unboundedness, letter channel
    unboundedness, non-termination and repeated control state
    reachability are all \NP{}-hard.
\end{proof}
Hence we deduce the main result of this Section.

\begin{thm}[Most properties are \NP{}-complete]
For flat FIFO machines, reachability, repeated reachability, repeated control-state reachability, termination, boundedness, channel-boundedness and letter-channel-boundedness are \NP{}-complete. Cyclicity can be decided in linear time.
\end{thm}

%% file: front-lossy.tex
Let us informally recall that lossy FIFO machines (often called lossy channel systems~\cite{DBLP:conf/cav/AbdullaBJ98}) are like FIFO machines except that the FIFO semantics allows the lost of any message in the FIFO channels in any configuration. The reachability set on each channel is then downward closed for the subword ordering, hence by Higman's Theorem, one deduces that it is regular; the knowledge of the regularity of the reachability set provides a semi-algorithm for deciding non-reachability by enumerating all inductive forward reachability invariants. By combining this semi-algorithm with a fair exploration of the reachability tree that enumerates all reachable configurations, one obtains an algorithm that decides reachability. Since reachability is Hyper-Ackermann-complete for lossy FIFO machines~\cite{DBLP:journals/toct/Schmitz16}, it is natural to study the complexity of reachability in \emph{flat} lossy FIFO machines.

%
%
%

Abdulla, Collomb-Annichini,  Bouajjani  and Jonsson studied the
verification of lossy FIFO machines by accelerating loops and
representing them by a class of regular expressions called Simple
Regular Expressions (SRE)~\cite{DBLP:conf/cav/AbdullaBJ98,DBLP:journals/fmsd/AbdullaCBJ04}.
Recall that SRE are exactly regular languages that are downward closed
(for the subword ordering). Suppose a lossy FIFO machine (with one channel to simplify notations) has a loop
labeled by $\sigma$ and $L$ is a SRE\@. Let $\sigma^*(L)$ denote the set
of channel contents reachable after executing the loop arbitrarily many
times, starting from channel contents that are in $L$. By analyzing the polynomial (quadratic)
algorithms for computing $\sigma^*(L)$, we obtain an upper bound
for the computation of the reachability set of a flat lossy FIFO
machine.

\begin{thm}\label{flatlossy}
The reachability set of a flat lossy FIFO machine $S$ is a SRE
that can be computed in exponential time.
\end{thm}
\begin{proof}
The SRE for $\sigma^*(L)$ where $L$ is a SRE can be computed in
quadratic time~\cite[Corollary~3]{DBLP:conf/cav/AbdullaBJ98}. By
iterating this computation on the flat structure, we obtain a SRE of
size exponential describing the reachability set.
\end{proof}
%
%

From Theorem~\ref{flatlossy}, we may deduce that reachability is in EXPTIME for flat lossy FIFO machines.  Since there is a linear algorithm for checking whether a SRE is
included in another one~\cite[Lemma 3]{DBLP:conf/cav/AbdullaBJ98}, we may use this algorithm for checking whether a word
$w=w_1w_2 \cdots w_n$ (where all $w_i$ are letters) is in a SRE $L$ by
testing whether the associated SRE $L_w =
(w_1+\epsilon).(w_2+\epsilon) \ldots (w_n+\epsilon)$ is included in $L$.
This proves that reachability of a configuration $(q,w)$ is in EXPTIME\@.
Let us show now that reachability is in  \NP{} for flat lossy FIFO machines.

We first prove that the control-state reachability problem is in \NP{} for front-lossy FIFO machines and then we will use an easy reduction of reachability in flat lossy FIFO machines to the control-state reachability problem for front-lossy FIFO machines.

A FIFO machine is said to be \emph{front-lossy} if at any time, any
letter at the front of any channel can be lost. A front-lossy FIFO
machine $S$ is a tuple $(Q,\channels,M,\Delta)$ as for standard FIFO
machines defined in Definition~\ref{def:fifoSystems}. Only the
semantics change for front-lossy machines. Suppose $w \in M^{*}$ is a
sequence and $\mathtt{c}$ be a channel; let ${(w)}_{\mathtt{c}}$ denote
the channel valuation that assigns $w$ to $\mathtt{c}$ and $\epsilon$
to all other channels. The front-lossy semantics is given by a
transition system $T_{S}$ as for standard semantics. For every transition
$q \fstrans{\mathtt{c}?a} q'$ of a front-lossy FIFO machine, the channel
valuation ${(wa)}_{\mathtt{c}}\cdot \fval$ results in the transition
$(q,{(wa)}_{\mathtt{c}}\cdot \fval) \fstrans{\mathtt{c}?a} (q',\fval)$
in $T_{S}$. Every transition $q \fstrans{\mathtt{c}!a} q'$
of $S$ and channel valuation $\fval \in {(M^{*})}^{\channels}$ results
in the transition $(q,\fval) \fstrans{\mathtt{c}!a} (q',\fval\cdot
\vec{a}_{\mathtt{c}})$ in $T_{S}$, as for standard FIFO machines.

We will prove that control state reachability in flat front-lossy FIFO
machines is in \NP{}, by adapting a construction from~\cite{EGM2012}.
We reproduce some definitions from~\cite{EGM2012} to be able to
describe our adaptation.
\begin{defiC}[\cite{EGM2012}]%
    \label{def:mhpda}
    A $d$-head pushdown automaton is a 9-tuple \[A=\langle S, \Sigma,
    \mathdollar, \Gamma, \Delta_{A}, \nu, s_{0}, \gamma_{0},
    S_{f}\rangle\]
    where
    \begin{enumerate}
        \item $S$ is a finite non-empty set of states,
        \item $\Sigma$ is the tape alphabet,
        \item $\mathdollar$ is a symbol not in $\Sigma$ (the endmarker
            for the tape),
        \item $\Gamma$ is the stack alphabet,
        \item $\Delta_{A}$, the set of transitions, is a mapping from
            $S \times (\Sigma \cup \set{\mathdollar} \cup
            \set{\epsilon}) \times \Gamma$ into finite subsets of
            $S \times \Gamma^{*}$,
        \item $\nu:S \to \set{1, \ldots, d}$ is the head selector
            function,
        \item $s_{0} \in S$ is the start state,
        \item $\gamma_{0} \in \Gamma$ is the initial pushdown symbol,
        \item $S_{f}\subseteq S$ is the set of final states.
    \end{enumerate}
\end{defiC}

\noindent
Intuitively, a $d$-HPDA has a a finite-state control ($S$), $d$
reading heads and a stack. All the reading heads read from the same
input tape. Each state $s \in S$ in the finite state control reads
from the head given by $\nu(s)$ and pops the top of the stack. The
transition relation then non-deterministically determines the new
control state and the sequence of symbols pushed on to the stack. The
read head moves one step to the right on the input tape. The size of a
$d$-HPDA is the number of bits needed to encode it, where the value
$\nu(s)$ is specified using binary encoding for every state $s$. We
write MHPDA for the family of $d$-HPDA for $d \ge 1$.

We write $(s,\gamma) \stackrel{\ket{\sigma}_{i}}{\hookrightarrow}
(s',w)$ whenever $(s',w) \in \Delta_{A}(s,\sigma,\gamma)$, where
$\nu(s)=i$. Let us fix a $d$-HPDA $A=\langle S, \Sigma, \mathdollar,
\Gamma, \Delta_{A}, \nu, s_{0}, \gamma_{0}, F\rangle$.  Define
$\mathcal{P}=\set{p:\set{1,\ldots,d} \to \Nat}$. An instantaneous
description (ID) of a is a triple $(s,\tau,p,w)\in S \times
\Sigma^{*}\mathdollar
\times \mathcal{P} \times \Gamma^{*}$. An ID $(s,\tau,p,w)$ denotes
that $A$ is in state $s$, the tape content is $\tau$, the reading head
$i$ is at the position $p(i)$ and the pushdown store content is $w$.
Let $\vdash$ be the binary relation between IDs defined as follows:
we have $(s,\tau,p,w\gamma) \vdash (s',\tau,p',ww')$ iff each of the
following conditions is satisfied:
\begin{enumerate}
    \item $(s,\gamma) \stackrel{\ket{\sigma}_{i}}{\hookrightarrow}
            (s',w')$ where $\nu(s)=i$ and the letter in position
            $p(i)$ of $\tau$ is $\sigma$.
    \item $p'(j)=p(j)+1$ if $j = \nu(s)$ and $p'(j)=p(j)$ otherwise.
\end{enumerate}
Let $\vdash^{*}$ denote the reflexive transitive closure of $\vdash$.

We say that reading head $i$ is off the tape in the ID $(s,\tau,p,w)$
if $p(i)$ is the last position of $\tau$, where the symbol is
$\mathdollar$. We say that $(s,\tau,p,w)$ is accepting iff $s \in
S_{F}$ and for every $i \in \set{1, \ldots, d}$, the reading head
$i$ is off the tape. A tape content $x\in \Sigma^{*}$ is accepted by
$A$ if $(s_{0},x\mathdollar,p_{0},\gamma_{0}) \vdash^{*}
(s,x\mathdollar,p,w)$ where $p_{0}(i)=1$ for all $i \in \set{1,
\ldots, d}$ and $(s,x\mathdollar,p,w)$ is some accepting ID\@. Let
$L(A)$ be the set of words in $\Sigma^{*}$ accepted by $A$.

A bounded expression is a regular expression
$\overline{w}=w_{1}^{*}\ldots w_{n}^{*}$, where each $w_{i}$ is a
non-empty word over $\Sigma$. With slight abuse of notation, we also
use $\overline{w}$ for the language defined by $\overline{w}$.

\begin{thmC}[\cite{EGM2012}]%
    \label{thm:mhpdaNonEmptiness}
    Let $\set{A_{i}}_{i \in \set{1, \ldots,q}}$ be a family of MHPDA such that
$A_{i}$ is a $d_{i}$-HPDA for every $i \in \set{1, \ldots, q}$. Let $d$ be a
constant such that $d\ge \max(d_{i} \mid i \in \set{1, \ldots,q})$. Let
$\overline{w}=w_{1}^{*} \ldots w_{n}^{*}$ be a bounded expression.  Checking
whether $\cap_{i=1}^{q}L(A_{i}) \cap \overline{w}$ is non-empty is in \NP\@.
\end{thmC}
We can now state our \NP{} upper bound result for flat front-lossy machines.
\begin{lem}%
    \label{lem:flatFrontLossyNP}
    The control state reachability problem for flat front-lossy machines is in \NP{}.
\end{lem}
\begin{proof}
Given a flat front-lossy machine $(Q,\channels,M,\Delta)$ with $p$
channels, an initial configuration $(q_0,\fval_0)$ and a target
control state $q$, we construct $(p+1)$ $2$-HPDAs $A_0,A_\mathtt{1},
\ldots, A_\mathtt{p}$ and a bounded exression $\overline{w}$ over
$\Delta$ such that some configuration $(q,\fval)$ is reachable from
$(q_0,\fval_0)$ iff $\cap_{i=1}^{q}L(A_{i}) \cap \overline{w}$ is
non-empty. We assume without loss of generality that
$\fval_{0}(\mathtt{c})=\epsilon$ for all channels $\mathtt{c}$; if
not, we prepend paths that add the required symbols to each channel.
Since the given front-lossy machine is flat, there is a bounded
expression $\overline{w}$ over $\Delta$ whose language is the set of
paths from $q_{0}$ to $q$; this is the bounded expression required.
Next we describe the $2$-HPDAs.

The automaton $A_0$ simply checks whether the tape content
is in the language of $\overline{w}$. This can be done with a single
reading head and without a stack. For every channel $\mathtt{c}$, the
2-HPDA $A_\mathtt{c}$ will check that the sequence of transitions in
its tape is viable, with respect to the contents of channel
$\mathtt{c}$. Suppose $M=\set{a_1, a_2, \ldots, a_n}$.
Figure~\ref{fig:frontLossy} illustrates $A_\mathtt{c}$, which uses the pushdown symbol
$\zeta$ and two reading heads $H$ and $h$. A transition labeled $\ket{\set{!}\times M}_H, \epsilon/\zeta$ is to be read as follows: the
reading head $H$ can read any transition $t$ in $\Delta$, provided $t$
sends some letter to channel $\mathtt{c}$, nothing is popped from the
stack and $\zeta$ is pushed on to the stack. At any state of
$A_\mathtt{c}$, there are self loops that can read any
transition in $\Delta$ that does not interact with channel $\mathtt{c}$. The self
loops are not shown in the figure to reduce clutter.

In state $q_{H}$, head $H$ reads all transitions in $\Delta$ that is
of the form $\mathtt{c}!a_{i}$ for any $i \in \set{1, \ldots, n}$,
until a transition of the form $\mathtt{c}?a_{i}$ for some $i \in
\set{1, \ldots, n}$ or $\mathdollar$ is read. When a transition of the
form $\mathtt{c}?a_{i}$ is read, control jumps to $q_{h}^{i}$. In
$q_{h}^{i}$, the head $h$ looks for the symbol
$\mathtt{c}!a_{i}$, skipping symbols of the form
$\mathtt{c}?a_{j}$ or $\mathtt{c}!a_{j}$. Intuitively,
if a message $a_{i}$ is to be retrieved from channel $\mathtt{c}$, it
should have been sent previously. The head $h$ looks for the
transition that sent $a_{i}$ previously. It skips over any
$\mathtt{c}?a_{j}$ since they don't denote send action. It skips over
$\mathtt{c}!a_{j}$ since in the front-lossy semantics, a message at
the front of the queue can be lost. Non-deterministcally, zero or more
occurrences of $\mathtt{c}!a_{j}$ are skipped; then
$\mathtt{c}!a_{i}$ is read and the control is back to $q_{H}$. To
ensure that the head $h$ does not move beyond $H$, we use the stack.
Whenever $H$ moves, it pushes $\zeta$ onto the stack and whenever $h$
moves, it pops $\zeta$ from the stack. In any reachable configuration
not in the state $q_f$, $A_{\mathtt{c}}$ maintains the invariant that
the number of symbols between $h$ and $H$ is equal to the number of
$\zeta$'s on the stack.  Because of this invariant, head $H$ will be
the first to read $\mathdollar$ in which case the control is updated
to $q_{f}$. Then all the remaining symbols are read by head $h$ until
it also reads $\mathdollar$. This way, every retrieve action is
matched with a unique send action, so the sequence of actions in the
tape of $A_{\mathtt{c}}$ is viable with respect to channel
$\mathtt{c}$. Since this is done for every channel, we infer that the
control state $q$ is reachable from $(q_{0},\fval_{0})$ iff
$\cap_{i=1}^{q}L(A_{i}) \cap \overline{w}$ is non-empty. Since
checking this later conditions is in \NP{} (by
Theorem~\ref{thm:mhpdaNonEmptiness}), we conclude that the
control state reachability problem for flat front-lossy systems is in
\NP\@.
\end{proof}
There is only a small difference between the construction given in the proof of
Lemma~\ref{lem:flatFrontLossyNP} and the one given in Section~IX of~\cite{EGM2012}. In our constrction, the self loops on
$q_{h}^{1}, \ldots, q_{h}^{n}$ in Figure~\ref{fig:frontLossy} can read
$\set{?}\times M$ as well as $\set{!}\times M$, whereas as in~\cite{EGM2012}, only $\set{?} \times M$ can be read.\\
\begin{figure}
\centering
\begin{tikzpicture}
    \node[circle, draw=black] (qH) at (0cm,0cm) {$q_H$};
    \node[circle, draw=black] (q1) at ([xshift=-4cm,yshift=-3cm]qH) {$q_{h}^{1}$};
    \node[circle, draw=black] (q2) at ([yshift=-4cm]qH) {$q_{h}^{2}$};
    \node[circle, draw=black] (q3) at ([xshift=4cm,yshift=-3cm]qH) {$q_{h}^{n}$};
    \node[circle, draw=black, double] (qf) at ([xshift=4cm]qH) {$q_{f}$};

    \draw[->] (qH) edge[loop above] node {$[\set{!}\times M\rangle_H, \epsilon/\zeta$} (qH);
    \draw[->] (qH) --node[auto=left,sloped, pos=0.8] {$[?a_1\rangle_H,\epsilon/\zeta$} (q1);
    \draw[->] (q1) edge[bend left] node[auto=left, rotate=45,pos=0.7, yshift=0.3cm] {$[!a_1\rangle_H,\zeta/\epsilon$} (qH);
    \draw[->] (q1) edge[loop below] node[auto=left] {$[\set{?,!} \times M\rangle_h,\zeta/\epsilon$} (q1);

    \draw[->] (qH) --node[auto=left,rotate=90,pos=0.9, yshift=-0.25cm] {$[?a_2\rangle_H,\epsilon/\zeta$} (q2);
    \draw[->] (q2) edge[bend left] node[auto=left, rotate=90,pos=0.7] {$[!a_2\rangle_H,\zeta/\epsilon$} (qH);
    \draw[->] (q2) edge[loop below] node[auto=left] {$[\set{?,!} \times M\rangle_h,\zeta/\epsilon$} (q2);

    \draw[->] (qH) --node[auto=left,sloped, pos=0.4] {$[?a_n\rangle_H,\epsilon/\zeta$} (q3);
    \draw[->] (q3) edge[bend left] node[auto=left, rotate=-40,pos=0.2, yshift=-0.3cm] {$[!a_n\rangle_H,\zeta/\epsilon$} (qH);
    \draw[->] (q3) edge[loop below] node[auto=left] {$[\set{?,!} \times M\rangle_h,\zeta/\epsilon$} (q3);

    \draw[->] (qH) --node[pos=0.5, auto=left] {$[\mathdollar\rangle_{H}$} (qf);
    \draw[->] (qf) edge[loop above] node {$[\set{?,!} \times M\rangle_{h}$} (qf);
    \draw[->] (qf) edge[loop below] node {$[\mathdollar\rangle_{h}$} (qf);

    \draw[->] ([xshift=-1cm,yshift=0.3cm]qH.center) -- (qH);

    \draw[loosely dotted] (q2) -- (q3);
\end{tikzpicture}
\caption{The 2-HPDA $A_\mathtt{c}$}%
\label{fig:frontLossy}
\end{figure}
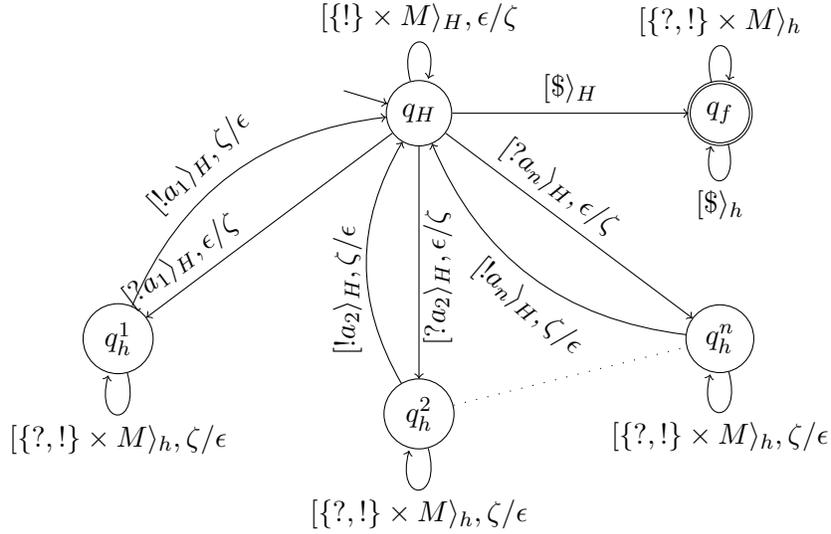

Hence we deduce that reachability is in  \NP{} for both flat
front-lossy FIFO machines and flat lossy FIFO machines. To achieve
these results, we reduce reachability (in both models) to
control-state reachability in a front-lossy FIFO machine.

To test
whether $(q,\fval)$ is reachable from $(q_0,\fval_0)$ in a flat
front-lossy FIFO machine $S$, we complete $S$ into the front-lossy FIFO machine $S_{(q,\fval)}$
by adding a new path in $S$ (in a similar way as the added path in
the proof of Proposition~\ref{reachability}), from $q$ to
$q_\mathrm{stop}$,
that essentially consumes $\fval$. We obtain that $(q,\fval)$ is
reachable in $S$ iff $q_\mathrm{stop}$ is reachable in the front-lossy FIFO
machine $S_{(q,\fval)}$.

To test whether $(q,\fval)$ is reachable from
$(q_0,\fval_0)$ in a flat lossy FIFO machine $S'$, we complete $S'$
into the lossy FIFO machine $S'_{(q,\fval)}$ by adding a new path in
$S'$ similarly as above from $q$ to $q_{\mathrm{stop}}$, such that the
added path essentially consumes $\fval$. We obtain that $(q,\fval)$ is
reachable in $S'$ iff $q_\mathrm{stop}$ is reachable in the flat lossy FIFO
machine $S'_{(q,\fval)}$. Then we observe that control state
reachability in flat lossy machines reduces to control state
reachability in flat front-lossy machines. Hence, reachability in flat
lossy FIFO machines reduces to control state reachability in flat
front-lossy FIFO machines, which is in \NP{}.

Finally, we can use the same reduction as
the one used in the proof of Theorem~\ref{lem:reachNPHard} to prove
\NP{}-hardness in flat (front-)lossy FIFO machines. The only way to 
reach the target state
in the FIFO machine of that reduction is to not have any losses in the
entire operation of the machine, in addition to satisfying the given
3-CNF formula. Hence, the introduction of (front-)lossy semantics will 
not change anything in the proof of Theorem~\ref{lem:reachNPHard}. So
we may deduce the following.

\begin{thm}
Reachability is \NP{}-complete for both flat front-lossy FIFO machines and flat lossy FIFO machines.
\end{thm}

During the review process, we were aware of a new unpublished paper from Schnoebelen~\cite{DBLP:journals/corr/abs-2007-05269} about flat lossy machines. This paper analyses iterations of lossy channel actions. As a consequence, it also proves the NP upper bound (for reachability and similar problems) by using compressed word techniques.


%% file: counter.tex
Suppose we want to model check flat FIFO machines against logics in
which atomic formulas are of the form $\#^{a}_{\mathtt{c}} \ge k$, which
means there are at least $k$ occurrences of the letter $a$ in channel
$\mathtt{c}$. Suppose the letter $a$ denotes an undesirable situation
and we would like to ensure that if there are $4$ occurrences of the
letter $a$, then the number reduces within the next two steps. This
is expressed by the LTL formula $\mathtt{G}(\#^{a}_{\mathtt{c}} \ge 4
\Rightarrow \mathtt{XX}\lnot(\#^{a}_{\mathtt{c}} \ge 4))$, where
$\mathtt{G}$ and $\mathtt{X}$ are the usual LTL operators.

There is no easy way of designing an algorithm for this model
checking problem based on the construction in~\cite{EGM2012}, even
though we solved reachability and related problems in previous
sections using that
construction. That
construction is based on simulating FIFO machines using automata that
have multiple reading heads on an input tape. The channel
contents of the FIFO machine are represented in the automaton as the
sequence of letters on the tape between two reading heads. There is no way in the
automaton to access
the tape contents between two heads, and hence no way to check the
number of occurrences of a specific letter in a channel.
CQDDs
introduced in~\cite{DBLP:journals/tcs/BouajjaniH99} represent the
entire
set of reachable states and they are also not suitable for model checking.

To overcome this problem, we introduce here a counter system to
simulate flat FIFO machines. This has the additional advantage of being
amenable to analysis using existing tools on counter machines.
Counter machines are finite state
automata augmented with counters that can store natural numbers.
Let $\counters$ be a finite set of counters and let \emph{guards over
$\counters$} be the set $G(\counters)$ of positive Boolean
combinations\footnote{In the literature, counter machines can have more
complicated guards, such as Presburger constraints. For our purposes,
this restricted version suffices.} of constraints of the form $C=0$ and $C>0$, where
$C \in \counters$.
\begin{defi}[Counter machines]
A counter machine $S$ is a tuple $( Q, \counters, \Delta )$
where $Q$ is a finite set of control states and $\Delta \subseteq Q
\times G(\counters) \times {\{-1,0,1\}}^\counters \times Q$ is a finite set of
transitions.
\end{defi}
%
We may add one or two labeling functions to the
tuple $( Q, \counters, \Delta )$ to denote labeled counter
machines. The semantics of a counter machine is a transition system with
set of states $Q \times \Nat^{\counters}$, called
\emph{configurations of the counter machines}. A counter
valuation $\cval \in \Nat^{\counters}$ \emph{satisfies} a guard $C=0$
(resp.~$C>0$) if $\cval(C)=0$ (resp.~$\cval(C)>0$), written as $\cval
\models C=0$ (resp.~$\cval\models C>0$). The satisfaction relation is
extended to Boolean combinations in the standard way. For every
transition $\delta=q\cstrans{\vec{u}}{g}q'$ in the counter machine, we
have transitions $(q,\cval_{1}) \xlongrightarrow{\delta}
(q',\cval_{2})$ in the associated transition system for every
$\cval_{1}$ such that $\cval_{1} \models g$ and $\cval_{2} =
\cval_{1}+\vec{u}$ (addition of vectors is done component-wise). We
write a transition $(q,C_{2}=0,\langle 1,0\rangle, q')$ as $q
\cstrans{C_{1}^{++}}{C_{2}=0} q'$, denoting addition of $1$ to $C_{1}$
by $C_{1}^{++}$. We denote by $\cstrans{}{}$ the union $\cup_{\delta
\in \Delta} \xlongrightarrow{\delta}$. A \emph{run} of the counter
machine is a finite or infinite sequence $(q_{0},\cval_{0})
\cstrans{}{} (q_{1}, \cval_{1}) \cstrans{}{} \cdots$ of
configurations, where each pair of consecutive configurations is in
the transition relation.

We assume for convenience that the message alphabet $M$ of a FIFO
machine is the
disjoint union of $M_{\mathtt{1}}, \ldots, M_{\mathtt{p}}$, where
$M_{\mathtt{c}}$ is the alphabet
for channel $\mathtt{c}$.
In the following, let $S=(Q,\channels,M, \Delta)$ be a flat FIFO
machine, where
the set of channels $\channels=\set{1, \ldots, \mathtt{p}}$ and the set of
    transitions $\Delta = \set{t_{1}, \ldots, t_{r}}$.

\paragraph{The counting abstraction machine}

The idea behind the \emph{counting abstraction machine} is to \emph{ignore the order
of letters} stored in the channels and use counters to remember only
the number of occurrences of each letter. If a transition $t$ sends
letter $a$, the corresponding transition in the counting abstraction
machine increments the counter $(a,t)$. If a transition
$t$ retrieves a letter $a$, the retrieved letter would have been
produced by some earlier transition $t'$; the corresponding
transition in the counting abstraction machine will decrement the
counter $(a,t')$. The counting abstraction machine doesn't exactly
simulate the flat FIFO machine. For example, if the transition labeled
${(a,t_{1})}^{--}$ in Fig.~\ref{fig:countingAbstraction} is
executed, we know that there is at least one occurrence of the letter
$a$ in the channel, since the counter $(a,t_{1})$ is greater than
zero at the beginning of the transition. However, it is not clear that
the letter $a$ is at the front of the channel; there might be an
occurrence of the letter $b$ at the front. This condition can't be
tested using  the counting abstraction machine. We use other counter
machines to maintain the
order of letters.

\medskip

Formally, the
    \emph{counting abstraction machine} corresponding to $S$ is a
    labeled counter
    machine $S_{\cnt} = (Q,\counters,\Delta_{\cnt}, \psi, T)$, where
    $(Q,\counters,\Delta_{\cnt})$ is a counter machine and $\psi, T$
    are labeling functions.
    The set of \emph{counters} $\counters$ is in bijection with $M
    \times \Delta$ and a counter will be denoted $c_{a,t}$ or shortly
    $(a,t)$, for $a \in M$ and $t \in \Delta$. The set $\Delta_{\cnt}$ of \emph{transitions} of
    $S_{\cnt}$ and the labeling functions  $\psi:\Delta_{\cnt} \to (M
    \times \Delta) \cup \set{\tau}$ and $T:\Delta_{\cnt}
    \to \Delta$ are defined as follows: for every transition $t \in \Delta$, one adds the following transitions in $\Delta_{\cnt}$:
 \begin{itemize}
        \item If $t$ \emph{sends} a message,
            $t=q_{1}\fstrans{\mathtt{c}!a}q_{2}$, then the transition
            $t_{\cnt}=q_{1}\cstrans{ {(a,t)}^{++}}{}q_{2}$ is added to $\Delta_{\cnt}$;
            we define $\psi(t_{\cnt})=\tau$ and $T(t_{\cnt})=t$.
        \item If
            $t=q_{1}\fstrans{}q_{2}$ doesn't change any channel content, then the transition
            $t_{\cnt}=q_{1}\cstrans{}{}q_{2}$ is added to $\Delta_{\cnt}$;
            we define $\psi(t_{\cnt})=\tau$ and $T(t_{\cnt})=t$.
        \item If $t$ \emph{receives} a message,
            $t=q_{1}\fstrans{\mathtt{c}?a}q_{2}$, then the set of transitions $A_t$ is added to $\Delta_{\cnt}$ with
            $A_t = \set{\delta_{a,t'}  =
            q_{1}\cstrans{{(a,t')}^{--}}{(a,t')>0}q_{2}
            \mid
            t' \text{ sends } a \text{ to channel } \mathtt{c}}$.
We define  $\psi(\delta_{a,t'})=(a,t')$ and $T(\delta_{a,t'})=t$, for all $\delta_{a,t'} \in A_t$.
    \end{itemize}

\noindent
The function $\psi$ above will be used for synchronization
 with other
counter machines later and $T$ will be used to match the traces of this
counter machine with those of the original flat FIFO machine. In
figures, we do not show the labels given by $\psi$ and
$T$. They can be easily determined. For a
transition $\delta_{a,t'} \in \Delta_{\cnt}$, it decrements the
counter $(a,t')$ and $\psi(\delta_{a,t'})=(a,t')$. Transitions
that don't decrement any counter are mapped to
$\tau$ by $\psi$.

\begin{exa}%
    \label{ex:countingAbstraction}
    Figure~\ref{fig:flatFIFOSystem} shows a flat FIFO machine and
Fig.~\ref{fig:countingAbstraction} shows its counting abstraction
machine.
\end{exa}
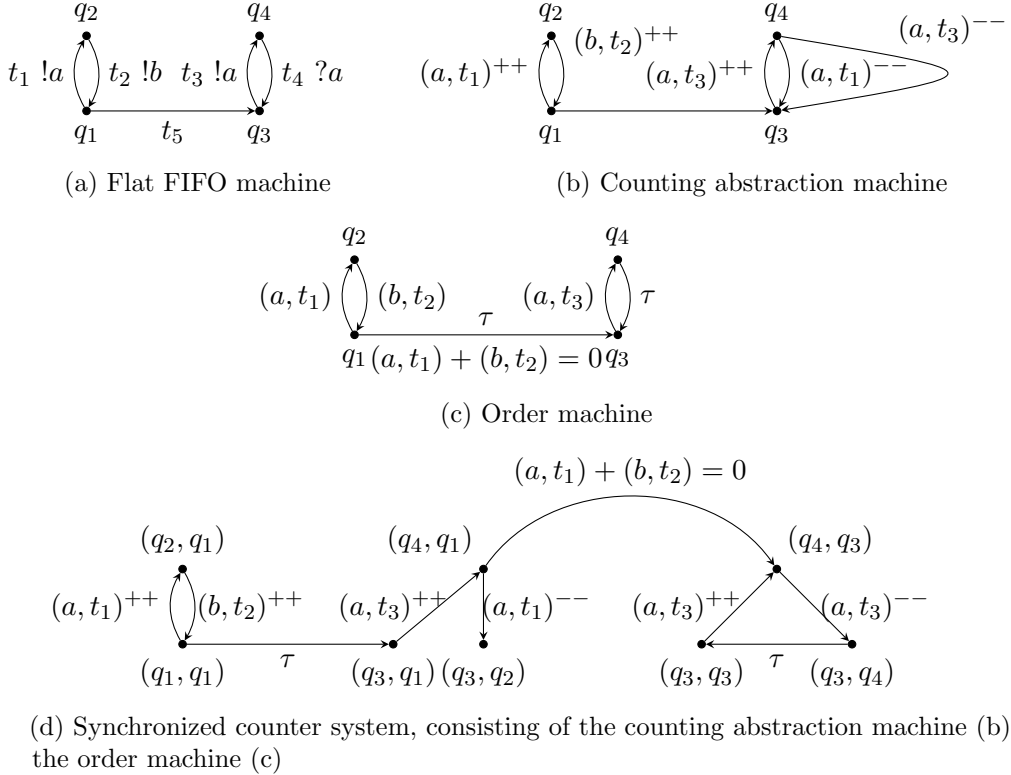
\begin{figure}[!htp]
    \centering
    \begin{subfigure}[t]{0.35\textwidth}
        \begin{tikzpicture}[>=stealth]
            \node[state, label=-90:$q_{1}$] (q1) at (0cm,0cm) {};
            \node[state, label=90:$q_{2}$] (q2) at ([yshift=1cm]q1) {};
            \node[state, label=-90:$q_{3}$] (q3) at ([xshift=2.3cm]q1) {};
            \node[state, label=90:$q_{4}$] (q4) at ([yshift=1cm]q3) {};

            \draw[->] (q1) to [bend left] node[auto=left] {$t_{1}~!a$} (q2);
            \draw[->] (q2) to [bend left] node[auto=left] {$t_{2}~!b$} (q1);
            \draw[->] (q1) to node[auto=right] {$t_{5}$} (q3);
            \draw[->] (q3) to [bend left] node[auto=left] {$t_{3}~!a$} (q4);
            \draw[->] (q4) to [bend left] node[auto=left] {$t_{4}~?a$} (q3);
        \end{tikzpicture}
        \caption{Flat FIFO machine}%
        \label{fig:flatFIFOSystem}
    \end{subfigure}
    \begin{subfigure}[t]{0.6\textwidth}
        \begin{tikzpicture}[>=stealth]
            \node[state, label=-90:$q_{1}$] (q1) at (0cm,0cm) {};
            \node[state, label=90:$q_{2}$] (q2) at ([yshift=1cm]q1) {};
            \node[state, label=-90:$q_{3}$] (q3) at ([xshift=3cm]q1) {};
            \node[state, label=90:$q_{4}$] (q4) at ([yshift=1cm]q3) {};

            \draw[->] (q1) to [bend left] node[auto=left] {$(a,t_{1})^{++}$} (q2);
            \draw[->] (q2) to [bend left] node[auto=left, pos=0.4] {$(b,t_{2})^{++}$} (q1);
            \draw[->] (q1) to (q3);
            \draw[->] (q3) to [bend left] node[auto=left, pos=0.47] {$(a,t_{3})^{++}$} (q4);
            \draw[->] (q4) to [bend left] node[auto=left] {$(a,t_{1})^{--}$} (q3);
            \path[->, draw=black] (q4) .. controls ([xshift=3cm]barycentric cs:q3=1,q4=1) ..
            node[auto=left, pos=0.2] {$(a,t_{3})^{--}$} (q3);
        \end{tikzpicture}
        \caption{Counting abstraction machine}%
        \label{fig:countingAbstraction}
    \end{subfigure}

    \begin{subfigure}[t]{0.6\textwidth}
        \begin{tikzpicture}[>=stealth]
            \node[state, label=-90:$q_{1}$] (q1) at (0cm,0cm) {};
            \node[state, label=90:$q_{2}$] (q2) at ([yshift=1cm]q1) {};
            \node[state, label=-90:$q_{3}$] (q3) at ([xshift=3.5cm]q1) {};
            \node[state, label=90:$q_{4}$] (q4) at ([yshift=1cm]q3) {};

            \draw[->] (q1) to [bend left] node[auto=left] {$(a,t_1)$} (q2);
            \draw[->] (q2) to [bend left] node[auto=left] {$(b,t_2)$} (q1);
            \draw[->] (q1) to node[auto=right] {$(a,t_{1}) + (b,t_{2}) = 0$} node[auto=left]{$\tau$} (q3);
            \draw[->] (q3) to [bend left] node[auto=left] {$(a,t_3)$} (q4);
            \draw[->] (q4) to [bend left] node[auto=left] {$\tau$} (q3);
            \path[use as bounding box] ([xshift=-2cm,yshift=-0.5cm]q1) rectangle
            ([xshift=1cm,yshift=0.75cm]q4);
        \end{tikzpicture}
        \caption{Order machine}%
        \label{fig:OrderSystem}
    \end{subfigure}

    \begin{subfigure}[t]{0.9\textwidth}
        \begin{tikzpicture}[>=stealth]
            \node[state, label=-90:{$(q_{1},q_{1})$}] (q1) at (0cm,0cm) {};
            \node[state, label=90:{$(q_{2},q_{1})$}] (q2) at ([yshift=1cm]q1) {};
            \node[state, label=-90:{$(q_{3},q_{1})$}] (q3) at ([xshift=2.8cm]q1) {};
            \node[state, label=93:{$(q_{4},q_{1})$}] (q4) at ([yshift=1cm,xshift=1.2cm]q3) {};
            \node[state, label=-90:{$(q_{3},q_{2})$}] (q5) at ([xshift=1.2cm]q3) {};
            \node[state, label=87:{$(q_{4},q_{3})$}] (q6) at ([xshift=3.9cm]q4) {};
            \node[state, label=-90:{$(q_{3},q_{4})$}] (q7) at ([xshift=4.9cm]q5) {};
            \node[state, label=-90:{$(q_{3},q_{3})$}] (q8) at ([xshift=-2cm]q7) {};

            \draw[->] (q1) to [bend left] node[auto=left] {$(a,t_1)^{++}$} (q2);
            \draw[->] (q2) to [bend left]  (q1);
            \node at ([xshift=0.9cm]barycentric cs:q2=1,q1=1) {$(b,t_2)^{++}$};
            \draw[->] (q1) to node[auto=right]{$\tau$} (q3);
            \draw[->] (q3) --  (q4);
            \node at ([xshift=-0.6cm]barycentric cs:q3=1,q4=1) {$(a,t_3)^{++}$};
            \draw[->] (q4) -- (q5);
            \node at ([xshift=0.7cm]barycentric cs:q4=1,q5=1) {$(a,t_1)^{--}$};
            \draw[->] (q4) to [bend left=55] node[auto=left] {${(a,t_{1}) + (b,t_{2})=0}$} (q6);
            \draw[->] (q6) -- (q7);
            \node at ([xshift=0.8cm]barycentric cs:q6=1,q7=1) {${(a,t_{3})^{--}}$};
            \draw[->] (q7) -- node[auto=left] {$\tau$} (q8);
            \draw[->] (q8) --  (q6);
            \node at ([xshift=-0.7cm]barycentric cs:q8=1,q6=1) {${(a,t_{3})^{++}}$};

            \path[use as bounding box] ([xshift=-2cm,yshift=-0.6cm]q1)
            rectangle ([xshift=1cm,yshift=1.8cm]q7|-q6);
        \end{tikzpicture}
        \caption{Synchronized counter system, consisting of the counting abstraction machine (b) and the order machine (c)}%
        \label{fig:synchronizedCounterSystem}
    \end{subfigure}
    \caption{An example flat FIFO machine (a) and the equivalent counter
        system (d).}
\end{figure}


Note that the counting abstraction machine associated with the flat
FIFO machine is not flat. Indeed, the receiving transition
$t_4=q_{4}\fstrans{?a}q_{3}$ in the FIFO machine is ``translated'' into
two decrementation transitions $q_{4}\fstrans{{(a,t_1)}^{--}}q_{3}$ and
$q_{4}\fstrans{{(a,t_3)}^{--}}q_{3}$ in the counting abstraction machine; these transitions breake the flatness property by creating nested loops on $\{q_3,q_4\}$.

%

\paragraph{The order machine}
The \emph{order machine for channel $\mathtt{c}$} is a labeled counter
machine
    $S_{\ord}^{\mathtt{c}} =
    (Q,\counters,\Delta_{\ord}^{\mathtt{c}},\psi^{\mathtt{c}})$, where
    $(Q,\counters,\Delta_{\ord}^{\mathtt{c}})$ is a counter machine and
    $\psi^{\mathtt{c}}$ is a labeling function.
   The set of control states $Q$ and the set of counters $\counters$
   are the same as in the counting abstraction machine.  The set
   $\Delta_{\ord}^{\mathtt{c}}$ of \emph{transitions} of
   $S_{\ord}^{\mathtt{c}}$ and the \emph{labeling function}
   $\psi^{\mathtt{c}}: \Delta_{\ord}^{\mathtt{c}} \to (M\times \Delta)
   \cup \set{\tau}$ are defined as follows: for every $t \in \Delta$,
   one adds the following transitions in $\Delta_{\ord}^{\mathtt{c}}$:
  \begin{itemize}
      \item If $t = q_{1} \fstrans{\mathtt{c}!a}q_{2}$, one adds to $\Delta_{\ord}^{\mathtt{c}}$ the transition $t'= q_{1} \rightarrow
            q_{2}$ and $\psi^{\mathtt{c}}(t')=(a,t)$.
        \item If $t = q_{1} \fstrans{x} q_{2}$ where $x$ doesn't contain a sending operation (of a letter) to channel $\mathtt{c}$,
           one adds to $\Delta_{\ord}^{\mathtt{c}}$ the transition $t'= q_{1} \rightarrow
            q_{2}$ and $\psi^{\mathtt{c}}(t')=\tau$.
    \end{itemize}
    While adding the transitions above, if $t$ happens to be the first
    transition after and outside a loop in $S$, we add a guard to the
    transition $t'$ that we have given in the above two cases. Suppose
    $t$ is the first
            transition after and outside a loop, and the loop is labeled by
            $\sigma$. We add the following guard to the
                transition $t'$.
            \begin{align*}
                \sum_{\substack{t'' \text{ occurs in }
                \sigma\\ a \in M}} (a,t'') & = 0
            \end{align*}
This constraint ensures that all the letters produced by iterations of
$\sigma$ are retrieved before letters produced by later transitions.

Figure~\ref{fig:OrderSystem} shows the order machine corresponding to
the flat FIFO machine of Fig.~\ref{fig:flatFIFOSystem}.

\paragraph{The synchronized counter system}

 We will
synchronize the counting abstraction machine $S_{\cnt}$ with the order
machines ${(S_{\ord}^{\mathtt{c}})}_c$ by
\emph{rendez-vous} on transition labels.

Suppose that the machine $S_{\ord}^{\mathtt{c}}$ is in state $q_{2}$
as shown in Fig.~\ref{fig:OrderSystem} and the machine $S_{\cnt}$ is in state $q_4$, as shown in
Fig.~\ref{fig:countingAbstraction}.
The machine $S_{\ord}^{\mathtt{c}}$ is in state $q_2$ and the
only transition going out from $q_{2}$ is labeled by $(b,t_{2})$,
denoting the fact that the next letter to be retrieved from the channel
is $b$. 
The
machine $S_{\cnt}$ can't execute the transition labeled with
${(a,t_{1})}^{--}$ in this configuration, since its
$\psi$-label is $(a,t_{1})$ and hence it can't synchronize with the
machine $S_{\ord}^{\mathtt{c}}$, whose next transition is labeled with $(b,t_{2})$. The guard
$(a,t_{1}) + (b,t_{2}) = 0$ in the bottom transition in
Fig.~\ref{fig:OrderSystem} ensures that all occurrences of
letters produced by iterations of the first loop are retrieved before
those produced by the second loop.

In the following, the \emph{label of a transition} refers
to the image of that transition under the function $\psi$ (if the
transition is in the counting abstraction machine) or the function
$\psi^{\mathtt{c}}$ (if the transition is in the order machine for channel
$\mathtt{c}$).

\medskip
The
    \emph{synchronized counter system} $S_{\sync}=
    S_{\cnt} \mid \mid S_{\ord}^{1}  \mid \mid \cdots  \mid \mid
    S_{\ord}^{\mathtt{c}}  \mid \mid \cdots \mid \mid
    S_{\ord}^{\mathtt{p}}$ is the
    synchronized (by rendez-vous) product of the  \emph{counting
abstraction machine} $S_{\cnt}$
    and the  \emph{order machines} $S_{\ord}^{\mathtt{c}}$ for all channels $\mathtt{c} \in
    \set{1, \ldots, \mathtt{p}}$. All counter machines share the same set of
    counters $\counters$ and have disjoint copies of the set of
    control states $Q$, so the global control states of the synchronized counter system are
    tuples in $Q^{\mathtt{p}+1}$. Transitions labeled with $\tau$
    need not synchronize with others. Each transition
    labeled (by the function $\psi$ or $\psi^{\mathtt{c}}$ as
    explained above) with an element of $M\times \Delta$ should synchronize with exactly one
    other transition that is similarly labeled. We extend the labeling
function $T$ of $S_\cnt$ to $S_\sync$ as follows: if a transition $t$
of $S_\cnt$ participates in a transition $t_s$ of $S_\sync$, then
$T(t_s)=T(t)$. If no transition from $S_\cnt$ participates in $t_s$,
then $T(t_s)=\tau$ and we call $t_s$ a silent transition.

\medskip
Since we have assumed that the channel alphabets for different
channels are mutually disjoint, synchronizations can only happen
between the counting abstraction machine and one of the order machines.
For a global control state $\overline{q} \in Q^{\mathtt{p}+1}$,
$\overline{q}(0)$ denotes the local state of the counting
abstraction machine
and $\overline{q}(\mathtt{c})$ denotes the local state of the order
machine for
channel $\mathtt{c}$. The synchronized counter system maintains the channel
contents of the flat FIFO machine as explained next.

\paragraph{A weak bisimulation between the FIFO machine and the synchronized system}
We now explain that every reachable configuration
$(\overline{q},\cval)$ of $S_{\sync}$ corresponds to a unique
configuration $h(\overline{q},\cval)$ of the original FIFO
machine $S$. The corresponding configuration of $S$ is
$h(\overline{q},\cval) =(\overline{q}(0),
h_1(v_1),h_2(v_2),\ldots,h_\mathtt{p}(v_\mathtt{p}))$, where the words
$v_{\mathtt{c}} \in \Delta^{*}$ and morphisms $h_{\mathtt{c}}:
\Delta^{*} \to M^{*}$ are as follows. Fix a channel $\mathtt{c}$. Let $v_\mathtt{c} \in
\Delta^*$ be a word labelling a path in $S$ from
$\overline{q}(\mathtt{c})$ to $\overline{q}(0)$ such that
$\mathit{Parikh}(v_\mathtt{c})(t)=\cval\left( (a,t) \right)$ for every
transition $t \in \Delta$ that sends some letter $a$
to channel $\mathtt{c}$.
Now, define $h_\mathtt{c}(t)=a$ if $t$ sends some letter $a$ to
channel $\mathtt{c}$
and
$h_\mathtt{c}(t)=\epsilon$ otherwise. The word $h_{c}(v_\mathtt{c})$
is unique since $S$ is flat and so the set of traces of $S$,
interpreted as a language over the alphabet $\Delta$, is included in a
bounded language (recall that a bounded language is included in a language of the
form $w_{1}^{*}w_{2}^{*}\cdots w_{k}^{*}$).
Intuitively, the path $v_{\mathtt{c}}$ gives the order of letters in
channel $\mathtt{c}$ and the counters give the number of occurrences
of each letter. Let us denote by $R_{h,\sync}$ the relation $\set{(
    h((\overline{q},\cval)), (\overline{q},\cval)) \mid
    (\overline{q},\cval) \text{ is reachable in } S_{\sync}}$.

\begin{exa}%
    \label{ex:correspondence}
    Figure~\ref{fig:synchronizedCounterSystem} shows the reachable states of the
    synchronized counter system for the flat FIFO machine in
    Fig.~\ref{fig:flatFIFOSystem}. Initially, both the counting
    abstraction machine and the order machine are in state $q_{1}$, so
    the global state is $(q_{1},q_{1})$. Then the counting abstraction
    machine may execute the transition labeled ${(a,t_{1})}^{++}$ and go
    to state $q_{2}$ while the order machine stays in state
    $q_{1}$, resulting in the global state $(q_{2}, q_{1})$. Consider the global state
    $(q_{3},q_{2})$
and counter valuation $\cval$ with
    $\cval( (a,t_{1})) = 2, \cval( (b,t_{2}))=3$ and $\cval((a,t_{3}))
    = 1$.  Then, for the only channel $\mathtt{c} = \mathtt{1}$,
    $v_{\mathtt{c}}={t_{2}(t_{1}t_{2})}^{2}t_{5}t_{3}t_{4}$ and
    $h_{\mathtt{c}}(v_{\mathtt{c}})=b{(ab)}^{2}a$.
\end{exa}

Let us recall that a relation $R$ between the reachable
configurations of the FIFO machine $S$ and the synchronized counter
system $S_{\sync}$ is a \emph{weak bisimulation} if every pair
$((q,\fval),(\overline{q},\cval)) \in R$ satisfies the following
conditions: (1) for
every transition $(q,\fval) \fstrans{t} (q',\fval')$ in $S$, there is
a sequence $\sigma$ of transitions in
$S_{\sync}$ such that $T(\sigma)\in\tau^{*}t\tau^{*}$,
$(\overline{q},\cval) \cstrans{\sigma}{}
(\overline{q'},\cval')$ and
$( (q',\fval'), (\overline{q'}, \cval')) \in R$,
(2) for every transition $(\overline{q},\cval) \cstrans{t_{s}}{}
(\overline{q'},\cval')$ in $S_{\sync}$ with $T(t_{s})=\tau$,
$((q,\fval),(\overline{q'},\cval')) \in R$ and (3) for every
transition $(\overline{q},\cval) \cstrans{t_{s}}{}
(\overline{q'},\cval')$ in $S_{\sync}$ with $T(t_{s})=t \ne\tau$,
$(q,\fval) \fstrans{t} (q',\fval')$ is a transition in $S$ and $(
(q',\fval'), (\overline{q'}, \cval')) \in R$.
\begin{prop}%
    \label{lem:weakBisimulation}
    The relation $R_{h,\sync}$ is
        a weak bisimulation.
\end{prop}
\begin{proof}
    Suppose $( h((\overline{q},\cval)), (\overline{q},\cval)) \in
    R_{h,\sync}$, where $\overline{q}(0)=q$.  Suppose there is a transition
    $h((\overline{q},\cval)) \fstrans{t} (q',\fval')$ in $S$. We have
    $h((\overline{q},\cval))=(\overline{q}(0),\fval)$, where
    $\fval(\mathtt{c})=h_{\mathtt{c}}(v_{\mathtt{c}})$ for every
    channel $c$ and $v_\mathtt{c} \in \Delta^*$ is a word labelling a
    path in $S$ from $\overline{q}(\mathtt{c})$ to $\overline{q}(0)$
    such that $\mathit{Parikh}(v_\mathtt{c})(t)=\cval\left( (a,t)
    \right)$ for every transition $t \in \Delta$ that sends some
    letter to channel $\mathtt{c}$ (and $a$ is the letter that is sent
    by $t$).

We will prove condition (1) above for weak bisimulation
    by a case analysis, depending on the type of transition $t$.

    Case 1: transition $t$ is of the form $(q,\fval) \rightarrow
    (q',\fval)$. In
    $S_{sync}$, the counting abstraction machine executes the
    transition $q \rightarrow q'$ and the order machines do not perform
    any transition. Then $S_{\sync}$ is in the configuration
    $(q',\overline{q}(\mathtt{1}),\ldots,
    \overline{q}(\mathtt{p}),\cval)$, where $\mathtt{p}$ is the number
    of channels. For every channel
    $\mathtt{c}$, we get a word $v_{c}'$ labelling a path in
    $S$ from $\overline{q}(\mathtt{c})$ to $q'$
    such that $\mathit{Parikh}(v_\mathtt{c})(t')=\cval\left( (a,t')
    \right)$ for every transition $t' \in \Delta$ that sends some
    letter to channel $\mathtt{c}$ (and $a$ is the letter that is sent
    by $t'$) as follows: we simply append the transition $q
    \rightarrow q'$ to $v_{c}$. Hence,
    $h_{\mathtt{c}}(v_{c}')=h_{\mathtt{c}}(v_{c})$ and
    $((q',\fval), (q',\overline{q}(\mathtt{1}),\ldots,
    \overline{q}(\mathtt{p}),\cval)) \in R_{h,\sync}$.

    Case 2: transition $t$ is of the form $(q,\fval)
    \fstrans{\mathtt{c}!a}(q',\fval\cdot \vec{a}_{\mathtt{c}})$
    (recall that $\vec{a}_{\mathtt{c}}$ is the channel valuation that
    assigns $a$ to channel $\mathtt{c}$ and $\epsilon$ to all others).
    In $S_{\sync}$, the counting abstraction machine executes the
    transition $q \cstrans{{(a,t)}^{++}}{} q'$ and the order machines do
    not execute any transitions. Then $S_{\sync}$ is in the
    configuration $(q',\overline{q}(\mathtt{1}),\ldots,
    \overline{q}(\mathtt{p}),\cval')$, where $\cval'$ is obtained
    from $\cval$ by adding one to the counter $(a,t)$. For every
    channel $\mathtt{c'}$, we get a word $v_{c'}'$ labelling a path in
    $S$ from $\overline{q}(\mathtt{c'})$ to $q'$ such that
    $\mathit{Parikh}(v_\mathtt{c'}')(t')=\cval'\left( (a',t') \right)$
    for every transition $t' \in \Delta$ that sends some letter to
    channel $\mathtt{c'}$ (and $a'$ is the letter that is sent by $t'$)
    as follows: we simply append the transition $q
    \fstrans{\mathtt{c}!a} q'$ to $v_{\mathtt{c'}}$. Hence,
    $h_{\mathtt{c'}}(v_{c'}')=h_{\mathtt{c'}}(v_{c'})$ for $\mathtt{c'}
    \ne \mathtt{c}$ and
    $h_{\mathtt{c}}(v_{c}')=h_{\mathtt{c}}(v_{c})\cdot a$. Hence,
    $((q',\fval \cdot \vec{a}_{\mathtt{c}}), (q',\overline{q}(\mathtt{1}),\ldots,
    \overline{q}(\mathtt{p}),\cval')) \in R_{h,\sync}$.

    Case 3: transition $t$ is of the form
    $(q,\vec{a}_{\mathtt{c}} \cdot \fval) \fstrans{\mathtt{c}?a}
    (q',\fval)$. Since $(q,\vec{a}_{\mathtt{c}} \cdot \fval) =
    h(\overline{q},\cval)$, $(\vec{a}_{\mathtt{c}} \cdot \fval)
    (c)=h_{\mathtt{c}}(v_{\mathtt{c}})$, where $v_{\mathtt{c}}$ is a
    word labeling a path in
    $S$ from $\overline{q}(\mathtt{c})$ to $q$
    such that $\mathit{Parikh}(v_\mathtt{c})(t')=\cval\left( (a',t')
    \right)$ for every transition $t' \in \Delta$ that sends some
    letter to channel $\mathtt{c}$ (and $a'$ is the letter that is sent
    by $t'$). Hence, the first transition in $v_{\mathtt{c}}$ that
    sends a letter to channel $c$ is of the form $t'=q_{1}
    \fstrans{\mathtt{c} !a} q_{2}$ and $\cval( (a,t'))\ge 1$. In
    $S_{\sync}$, the order machine for channel $\mathtt{c}$ executes
    the sequence of transitions from $\overline{q}(\mathtt{c})$ to
    $q_{2}$; note that $\psi_{\mathtt{c}}$ labels the last transition
    of this sequence with $(a,t')$ and labels other transitions in
    this sequence with $\tau$. The order machines for other
    channels do not execute any transitions. The counting abstraction
machine
    executes the transition $q \cstrans{ {(a,t')}^{--}}{(a,t')>0} q'$,
    which is labeled by $\psi$ with $(a,t')$ so it can synchronize
    with the transition $q_{1} \fstrans{} q_{2}$ executed by the order
    machine for channel $c$. Now the synchronized counter system
    $S_{\sync}$ is in the configuration $( \overline{q'},\cval')$,
    where $\overline{q'}$ is obtained from $\overline{q}$ by changing
    $\overline{q}(0)$ from $q$ to $q'$ and changing
    $\overline{q}(\mathtt{c})$ to $q_{2}$ and $\cval'$ is obtained
    from $\cval$ by subtracting one from the counter $(a,t')$. For
    channels $\mathtt{c'} \ne \mathtt{c}$, let
    $v_{\mathtt{c'}}'=v_{\mathtt{c'}}$ and let $v_{\mathtt{c}}'$ be obtained
    from $v_{\mathtt{c}}$ by removing the prefix ending at
    $q_{2}$. Now for every channel $\mathtt{c'}$, the word $v_{c'}'$
    labels a path in
    $S$ from $\overline{q'}(\mathtt{c'})$ to $q'$ such that
    $\mathit{Parikh}(v_\mathtt{c})(t')=\cval'\left( (a',t') \right)$
    for every transition $t' \in \Delta$ that sends some letter to
    channel $\mathtt{c'}$ (and $a'$ is the letter that is sent by
    $t'$). Hence, $((q',\fval), (\overline{q'},\cval')) \in
    R_{h,\sync}$ (end of Case 3 and of condition (1)).

    Next we prove condition (2) for weak bisimulation: for every
    transition $(\overline{q},\cval) \cstrans{t_{s}}{}
    (\overline{q'},\cval')$ in $S_{\sync}$ with $T(t_{s})=\tau$, we
    will show that $((q,\fval),(\overline{q'},\cval')) \in
    R_{h,\sync}$. Recall that the labeling function $T$ of $S_\cnt$ is
    extended to $S_\sync$ as follows: if a transition $t$ of $S_\cnt$
    participates in a transition $t_s$ of $S_\sync$, then
    $T(t_s)=T(t)$. If no transition from $S_\cnt$ participates in
    $t_s$, then $T(t_s)=\tau$. Hence,
    if $S_{\sync}$ executes a transition
    $(\overline{q},\cval) \fstrans{t_{s}}
    (\overline{q'},\cval')$ and $T(t_{s})=\tau$, the counting
    abstraction machine does not participate in $t_{s}$. The only
    transition participating in $t_{s}$ is some transition $q_{1}
    \rightarrow q_{2}$ in $S_{\ord}^{\mathtt{c}}$ for some channel
    $\mathtt{c}$ satisfying the following property:
    $q_{1} \fstrans{x} q_{2}$ is a transition in the FIFO
    machine $S$ and $x$ does not contain any sending operation of any
    letter to $\mathtt{c}$. In this case, $S_{\sync}$ goes to the
    configuration $(\overline{q'},\cval')$ where $\cval'=\cval$ and
    $\overline{q'}$ is obtained from $\overline{q}$ by changing
    $\overline{q}(\mathtt{c})$ from $q_{1}$ to $q_{2}$. For channels
    $\mathtt{c'} \ne \mathtt{c}$, let
    $v_{\mathtt{c'}}'=v_{\mathtt{c'}}$ and let $v_{\mathtt{c}}'$ be
    obtained from $v_{\mathtt{c}}$ by removing the prefix transition
    $q_{1} \fstrans{x} q_{2}$. Now for every channel $\mathtt{c'}$,
    the word $v_{c'}'$ labels a path in $S$ from
    $\overline{q'}(\mathtt{c'})$ to $\overline{q}(0)$ such that
    $\mathit{Parikh}(v_\mathtt{c})(t')=\cval'\left( (a',t') \right)$
    for every transition $t' \in \Delta$ that sends some letter to
    channel $\mathtt{c'}$ (and $a'$ is the letter that is sent by
    $t'$). Hence, $(h((\overline{q},\cval)), (\overline{q'},\cval'))
    \in R_{h,\sync}$.

    Next we prove condition (3) for weak bisimulation by a case
    analysis, depending on the type of transition $t_{s}$.

    Case 1: the transition $t_{s}$ is of the form $t_{\cnt}=q
    \cstrans{}{} q'$ executed by the counting abstraction machine. Then
    $S_{\sync}$ goes to the configuration
    $(q',\overline{q}(\mathtt{1}),\ldots,
    \overline{q}(\mathtt{p}),\cval)$. If $h(
    (\overline{q},\cval))=(q,\fval)$, then the FIFO machine $S$
    executes the transition $q \fstrans{} q'$ and we conclude that
    $( (q',\fval), (q',\overline{q}(\mathtt{1}),\ldots,
    \overline{q}(\mathtt{p}),\cval)) \in R_{h,\sync}$ as in
    case 1 above.

    Case 2: the transition $t_{s}$ is of the form $t_{\cnt}=q
    \cstrans{ {(a,t)}^{++}}{} q'$ executed by the counting abstraction
    machine, where $t=q \fstrans{\mathtt{c}!a} q'$ is a transition of
    the FIFO machine $S$. Then $S_{\sync}$ goes to the configuration
    $(q',\overline{q}(\mathtt{1}),\ldots,
    \overline{q}(\mathtt{p}),\cval')$, where $\cval'$ is obtained from
    $\cval$ by adding one to the counter $(a,t)$. If $h(
    (\overline{q},\cval))=(q,\fval)$, then the FIFO machine $S$
    executes the transition $q \fstrans{\mathtt{c}!a} q'$ and goes to
    the configuration $(q,\fval \cdot \vec{a}_{\mathtt{c}})$. We
    conclude that $( (q,\fval \cdot \vec{a}_{\mathtt{c}}),
    (q',\overline{q}(\mathtt{1}),\ldots,
    \overline{q}(\mathtt{p}),\cval'))$ as in case 2 above.

    Case 3: the transition $t_{s}$ is a synchronized transition with
    the counting abstraction machine executing the transition
    $\delta_{a,t'}=q \cstrans{ {(a,t')}^{--}}{ (a,t')>0} q'$ (where $t'$
    sends the letter $a$ to channel $\mathtt{c}$) and the order
machine
    for channel $\mathtt{c}$ executing the transition $q_{1}
    \rightarrow q_{2}$ (which is labeled with $(a,t')$ by
    $\psi_{\mathtt{c}}$). The synchronized transition system
    $S_{\sync}$ goes to the configuration $( \overline{q'},\cval')$,
    where $\overline{q'}$ is obtained from $\overline{q}$ by changing
    $\overline{q}(0)$ from $q$ to $q'$ and changing
    $\overline{q}(\mathtt{c})$ to $q_{2}$ and $\cval'$ is obtained
    from $\cval$ by subtracting one from the counter $(a,t')$. If
    $h ( (\overline{q},\cval))=(q,\vec{a}_{\mathtt{c}}\cdot\fval)$,
    then the FIFO machine $S$ executes the transition
    $(q,\vec{a}_{\mathtt{c}} \cdot \fval) \fstrans{\mathtt{c}?a}
    (q',\fval)$. We conclude that $( (q',\fval), (
    \overline{q'},\cval')) \in R_{h,\sync}$ as in case 3
    above.
\end{proof}

\paragraph{A bisimulation between the FIFO machine and the modified synchronized system}

We proved weak bisimulation above instead of bisimulation, due to
the presence of silent transitions in the order machines participating in
$S_\sync$. We can modify the order machines as follows to get a bisimulation.
For every channel $\mathtt{c}$ and every transition $q_1
\cstrans{}{} q_2$ labeled $\tau$ in $S_\ord^\mathtt{c}$, remove the
transition and merge the two states $q_1, q_2$ into one state. If
exactly one
of the two states $q_1,q_2$ was an anchor state, retain the name of
the anchor state as the name of the merged state. Otherwise, retain
$q_2$ as the name of the merged state. Repeat this process
until there are no more transitions labeled $\tau$. Note that we have
only removed transitions that do not correspond to any transition of
$S$ sending letters to channel $\mathtt{c}$. Such transitions are
assigned $\epsilon$ by the morphism $h_\mathtt{c}$ defined in the
paragraph preceding Ex.~\ref{ex:correspondence}. Hence, the deletion
of $\tau$-labeled transitions do not affect the correspondence between
the configurations of $S$ and $S_\sync$. If there are no sending
transitions between two anchor states, the above deletion procedure
may result in two anchor states getting merged, destroying the
flatness of the order machine. Next we describe a way to tackle this.

Suppose a transition $t'$ in the order machine modified as above
corresponds to a transition $t$ in the original flat FIFO machine $S$.
Suppose this transition $t$ of $S$ is in a loop $\ell$, which is
labeled by the sequence of transitions $\sigma$. For every transition $t_1$
in $S$ outside $\ell$ but reachable from states in $\ell$, we make the
following modification.
If the order machine has a transition $t'_1$ corresponding to $t_1$, we
add the following guard to $t'_1$.
\begin{align*}
\sum_{\substack{t'' \text{ occurs in }
\sigma\\ a \in M}} (a,t'') & = 0
\end{align*}
These guards ensure that all letters sent by transitions in $\ell$ are
retrieved before retrieving letters sent by later transitions. In
addition, the guards ensure that the modified order machine is
flattable. Suppose the loop $\ell$ in $S$ corresponds to loop $\ell'$
in $S_\ord^\mathtt{c}$. If a transition occurring after and outside
the loop $\ell'$
is fired in $S_\ord^\mathtt{c}$, loop $\ell'$ can't be entered again. The
reason is that any transition $t''$ in the loop $\ell'$
tries to decrement some counter $(a,t'')$, but it can't be decremented
since it has value $0$, as checked in the guard newly added to every
transition occurring after $\ell'$.

The modified order machines don't have $\tau$-labeled transitions
anymore, hence the modified synchronized counter system $S'_\sync$
doesn't have silent transitions. Now a proof similar to that of
Proposition~\ref{lem:weakBisimulation} can be used to show bisimulation
between $S$ and the modified synchronized counter system $S'_\sync$.

Let $R'_{h,\sync}$ be  the relation $\set{(
        h((\overline{q},\cval)), (\overline{q},\cval)) \mid
        (\overline{q},\cval) \text{ is reachable in } S'_{\sync}}$.

\begin{prop}%
    \label{lem:Bisimulation}
    The relation $R'_{h,\sync}$ is
        a bisimulation.
\end{prop}

\paragraph{Trace-flattening}
The counting abstraction machine $S_{\cnt}$ is not flat in general. 
 E.g., there
are two transitions from $q_{4}$ to $q_{3}$ in
Fig.~\ref{fig:countingAbstraction}. Those two states are in more than
one loop, violating the condition of flatness.
\begin{figure}[!htp]
    \centering
        \begin{subfigure}[t]{0.4\textwidth}
            \begin{tikzpicture}[>=stealth]
                \node[state, label=-90:$q_{1}$] (q1) at (0cm,0cm) {};
                \node[state, label=90:$q_{2}$] (q2) at ([yshift=1cm]q1) {};
                \node[state, label=-90:$q_{3}$] (q3) at ([xshift=2.3cm]q1) {};
                \node[state, label=90:$q_{4}$] (q4) at ([yshift=1cm]q3) {};

                \draw[->] (q1) to [bend left] node[auto=left] {$t_{1}~!a$} (q2);
                \draw[->] (q2) to [bend left] node[auto=left] {$t_{2}~!b$} (q1);
                \draw[->] (q1) to node[auto=right] {$t_{5}$} (q3);
                \draw[->] (q3) to [bend left] node[auto=left] {$t_{3}~!a$} (q4);
                \draw[->] (q4) to [bend left] node[auto=left] {$t_{4}~?a$} (q3);
            \end{tikzpicture}
            \caption{Flat FIFO machine}
        \end{subfigure}
        \begin{subfigure}[t]{0.5\textwidth}
            \begin{tikzpicture}[>=stealth]
                \node[state, label={[black!20]-90:$q_{1}$}, unreachable] (q1) at (0cm,0cm) {};
                \node[state, label={[black!20]90:$q_{2}$}, unreachable] (q2) at ([yshift=1cm]q1) {};
                \node[state, label=-90:$q_{3}$] (q3) at ([xshift=3cm]q1) {};
                \node[state, label=90:$q_{4}$] (q4) at ([yshift=1cm]q3) {};

                \draw[->, draw=black!20] (q1) to [bend left] (q2);
                \draw[->, draw=black!20] (q2) to [bend left] (q1);
                \draw[->, unreachable] (q1) to (q3);
                \draw[->] (q3) to [bend left] node[auto=left, pos=0.47] {$(a,t_{3})^{++}$} (q4);
                \draw[->] (q4) to [bend left] node[auto=left] {$(a,t_{1})^{--}$} (q3);
                \path[->, draw=black] (q4) .. controls ([xshift=3cm]barycentric cs:q3=1,q4=1) ..
                node[auto=left, pos=0.2] {$(a,t_{3})^{--}$} (q3);
            \end{tikzpicture}
            \caption{Counting abstraction machine (grey part no longer reachable)}
        \end{subfigure}
        \begin{subfigure}[t]{0.4\textwidth}
            \begin{tikzpicture}[>=stealth,yshift=-2.3cm]
                \node[state, label={[black!20]-90:$q_{1}$}, unreachable] (q1) at (0cm,0cm) {};
                \node[state, label={[black!20]90:$q_{2}$}, unreachable] (q2) at ([yshift=1cm]q1) {};
                \node[state, label=-90:$q_{3}$] (q3) at ([xshift=3.5cm]q1) {};
                \node[state, label=90:$q_{4}$] (q4) at ([yshift=1cm]q3) {};

                \draw[->, draw=black!20] (q1) to [bend left] (q2);
                \draw[->, draw=black!20] (q2) to [bend left] (q1);
                \draw[->, unreachable] (q1) to (q3);
                \draw[->] (q3) to [bend left] node[auto=left] {$(a,t_3)$} (q4);
                \draw[->] (q4) to [bend left] node[auto=left] {$\tau$} (q3);

                \path[use as bounding box]
                ([xshift=-0.5cm,yshift=-0.5cm]q1) rectangle
                ([xshift=0.5cm,yshift=0.5cm]q4);
            \end{tikzpicture}
            \caption{Order machine (grey part\\ no longer reachable)}
        \end{subfigure}
        \begin{subfigure}[t]{0.4\textwidth}
            \begin{tikzpicture}[>=stealth]
                \node[state, label=90:{$(q_{4},q_{3})$}] (q6) at (0cm,0cm) {};
                \node[state, label=-90:{$(q_{3},q_{4})$}] (q7) at ([yshift=-1cm, xshift=1cm]q6) {};
                \node[state, label=-90:{$(q_{3},q_{3})$}] (q8) at ([xshift=-2cm]q7) {};

                \draw[->] (q6) -- node[auto=left] {${(a,t_{3})^{--}}$} (q7);
                \draw[->] (q7) -- node[auto=left] {$\tau$} (q8);
                \draw[->] (q8) -- node[auto=left, pos=0.2] {${(a,t_{3})^{++}}$} (q6);
                \path[use as bounding box]
                ([xshift=-1.5cm,yshift=-0.7cm]q8) rectangle
                ([xshift=1cm,yshift=1cm]q6-|q7);
            \end{tikzpicture}
            \caption{Part of synchronized counter system still reachable}%
            \label{fig:traceFlat}
        \end{subfigure}
        \caption{Flattening}%
    \label{fig:flattening}
\end{figure}
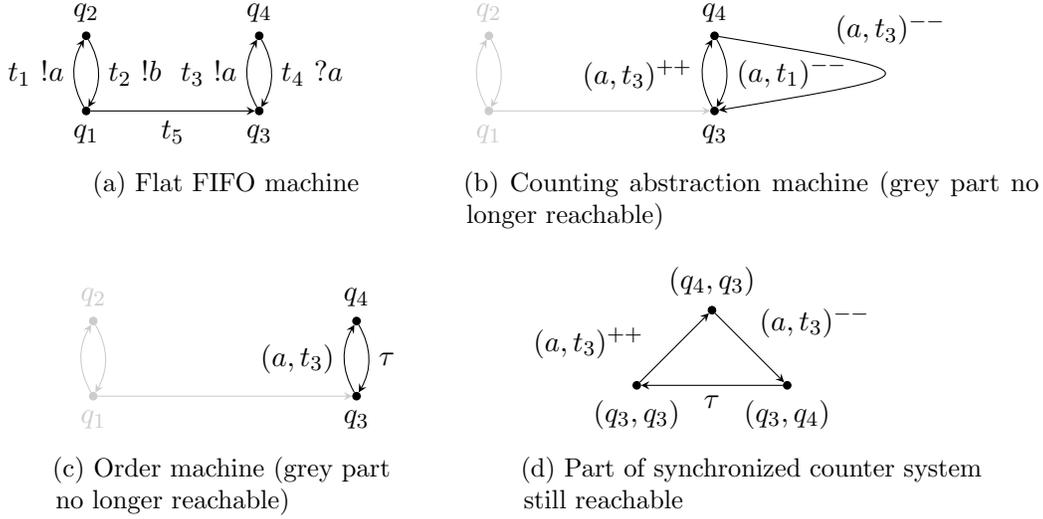
However, suppose a run is visiting states $q_{3},q_{4}$ of the
counting abstraction machine and states $q_{3}, q_{4}$ of the order
machine as shown in Fig.~\ref{fig:flattening} (parts of the system
that are no longer reachable are greyed out). Now the transition
labeled ${(a,t_{1})}^{--}$ can't be used and the run is as shown in
Fig.~\ref{fig:traceFlat}, which is a flat counter machine. In general,
suppose $\ell_{0}, \ell_1, \ldots, \ell_{r}$ are
the loops in $S$. There is a flat counter machine $S_{\flatcs}$ whose
set of runs is the set of runs $\rho$ of the synchronized
transition system which satisfy the following property: in $\rho$, all
local states of the counting abstraction machine are in some loop
$\ell_{i}$ and for every channel $\mathtt{c}$, all local states of the
order machine $S_{\ord}^{\mathtt{c}}$ are in some loop
$\ell_{\mathtt{c}}$. This is the intuition for the next result.

Let $\mathrm{traces}(S_{\sync})$ be the set of all runs
of $S_{\sync}$. Let $S'$ be another counter machine with set of
states $Q'$ and the same set of counters as $S_{\sync}$ and let
$f:Q' \to Q$ be a function. We say that $S'$ is a $f$-flattening
of $S_{\sync}$~\cite[Definition~6]{DFGV2006} if $S'$ is flat and for
every transition $q \cstrans{u}{g} q'$ of $S'$, $f(q) \cstrans{u}{g}
f(q')$ is a transition in $S_{\sync}$. Further, $S'$ is a $f$-trace-flattening
of $S_{\sync}$~\cite[Definition~8]{DFGV2006} if $S'$ is a $f$-flattening of
$S_{\sync}$ and $\mathrm{traces}(S_{\sync})=f(\mathrm{traces}(S'))$.
\begin{prop}%
    \label{lem:traceFlattable}
    The synchronized counter system $S_{\sync}$ is trace-flattable.
\end{prop}

\begin{proof}
    Starting from a global state $\overline{q}$ of $S_{\sync}$, we claim
    that we can build a flat counter machine that is a trace-flattening
    of $S_{\sync}$. Let $n_{0}$ be the number of loops in $S$ reachable
    from $\overline{q}(0)$. For each channel $\mathtt{c}$, let
    $n_{\mathtt{c}}$ be the number of loops in $S$ reachable from
    $\overline{q}(\mathtt{c})$. We prove the claim by induction on the
    vector $\langle n_{0}, n_{\mathtt{1}}, \ldots, n_{\mathtt{p}}
    \rangle$. The order on vectors is component-wise comparison ---
    $\langle n_{0}, n_{\mathtt{1}}, \ldots, n_{\mathtt{p}} \rangle <
    \langle n_{0}', n_{\mathtt{1}}', \ldots, n_{\mathtt{p}}' \rangle$
    if $n_{i} \le n_{i}'$ for all $i \in \set{0, \ldots, \mathtt{p}}$
    and $n_{j} < n_{j}'$ for some $j \in \set{0, \ldots, \mathtt{p}}$.

    For the base case, $\langle n_{0}, n_{\mathtt{1}}, \ldots, n_{\mathtt{p}}
    \rangle = \vec{0}$. From such a global state, the counting
    abstraction machine and order machines for all the channels have
    unique paths to follow and hence there is a unique run of
    $S_{\sync}$. This unique run can be easily simulated by a flat
    counter machine, proving the base case.

For the induction step, suppose $\ell_{0}$ is the first loop in
    $S$ reachable from $\overline{q}(0)$ and for every channel
    $\mathtt{c}$, suppose $\ell_{\mathtt{c}}$ is the
    first loop in $S$ reachable from $\overline{q}(\mathtt{c})$, with
$\ell'_c$ being the corresponding loop in $S_\ord^\mathtt{c}$. There
    is a flat counter machine $S_{\flatcs}$ described in the paragraph
    preceding this lemma, which can simulate runs of the synchronized
    counter system as long as the counting abstraction machinew doesn't
    exit the loop $\ell_{0}$ and for every channel $\mathtt{c}$, the
    order machine $S_{\ord}^{\mathtt{c}}$ doesn't exit the loop
    $\ell'_{\mathtt{c}}$. If the counting abstraction machine exits the
    loop $\ell_{0}$ (or the order machine $S_{\ord}^{\mathtt{c}}$ exits
    the loop $\ell'_{\mathtt{c}}$ for some channel $\mathtt{c}$), then
    the vector $\langle n_{0}-1,n_{\mathtt{1}}, \ldots,
    n_{\mathtt{p}}\rangle$ (or the vector $\langle n_{0},
    n_{\mathtt{1}}, \ldots, n_{\mathtt{c}}-1, \ldots,
    n_{\mathtt{p}}\rangle$) is strictly smaller than the vector $\langle
    n_{0}, n_{\mathtt{1}}, \ldots, n_{\mathtt{p}}\rangle$.\footnote{This step fails in non-flat FIFO machines; if a
    loop is exited in a non-flat FIFO machine, it may be possible to reach
the loop again, so the vector doesn't necessarily decrease.} The induction
    hypothesis shows that there is a flat counter machine $S_{\flatcs}'$
    that can cover the remaining possible runs. We sequentially
    compose $S_{\flatcs}$ and $S_{\flatcs}'$ by identifying the initial
    state of $S_{\flatcs}'$ with the state of $S_{\flatcs}$ in which the
    counting abstraction machine exits the loop $\ell_{0}$ (or the
    order machine $S_{\ord}^{\mathtt{c}}$ exits the loop
    $\ell'_{\mathtt{c}}$). There are finitely many possibilities of the
    counting abstraction machine or one of the order machines exiting a loop; for
    each of these possibilities, the induction hypothesis gives a
    flat counter machine $S_{\flatcs}'$. We sequentially compose
    $S_{\flatcs}$ with all such flat counter machines $S_{\flatcs}'$. The
    result is a trace-flattening of the synchronized counter system.
\end{proof}

Let $S_{\flatcs}$ be a trace-flattening of $S_{\sync}$. In general, the size of
$S_{\flatcs}$ is exponential in the size of $S_{\sync}$, which is exponential
in the size of $S$. In theory, problems on flat FIFO machines can be solved by
using tools on counter machines (bisimulation preserves CTL\textsuperscript{*}
and trace-flattening preserves LTL~\cite[Theorem1]{DFGV2006}); hence we deduce:

\begin{thm}%
\label{ltlctl}
LTL is decidable for flat FIFO machines.
\end{thm}

The decidability of CTL\textsuperscript{*} is an open problem for
bisimulation-flattable counter machines~\cite{DFGV2006}, so we cannot
use it for deciding CTL\textsuperscript{*} in flat FIFO machines.